\documentclass[a4paper,11pt]{article}
\pdfoutput=1 
\usepackage{jheppub} 
\usepackage[T1]{fontenc} 

\usepackage{indentfirst}
\usepackage{amsmath,amssymb,amsthm}
\usepackage{calligra,mathrsfs}
\usepackage{graphicx,color}
\usepackage[all,cmtip]{xy}
\usepackage{bbm}
\usepackage{hyperref}
\usepackage{mathtools}
\usepackage{verbatim}
\usepackage{stmaryrd}
\usepackage{tikz}
\usetikzlibrary{cd}
\usetikzlibrary{matrix}

\usepackage{bm}
\usepackage[margin=0.8cm,font=footnotesize]{caption}

\definecolor{lcolor}{rgb}{0.6,0.3,0.3}
\hypersetup{colorlinks=true,linkcolor=lcolor,urlcolor=lcolor,citecolor=lcolor}

\newcommand{\tc}[2]{#2}

\newcommand{\p}{\partial}

\newcommand{\R}{\mathbbm{R}}
\newcommand{\C}{\mathbbm{C}}

\newcommand{\eps}{\epsilon}
\newcommand{\step}{\vskip 3mm}

\DeclareMathOperator{\image}{image}

\DeclareMathOperator{\Res}{Res}

\renewcommand{\subset}{\subseteq}

\newtheoremstyle{mytheoremstyle}
{14pt}% space above
{14pt}% space below
{\itshape}% body font
{20pt}% indent amount
{\bfseries}% theorem head font
{.}% punctuation after theorem head
{.5em}% space after theorem head
{}% theorem head spec

\newtheoremstyle{myremarkstyle}
{10pt}% space above
{10pt}% space below
{}% body font
{20pt}% indent amount
{\itshape}% theorem head font
{.}% punctuation after theorem head
{.5em}% space after theorem head
{}% theorem head spec

\newtheoremstyle{myproofstyle}
{12pt}% space above
{12pt}% space below
{}% body font
{20pt}% indent amount
{\itshape}% theorem head font
{.}% punctuation after theorem head
{.5em}% space after theorem head
{}% theorem head spec

\theoremstyle{mytheoremstyle}
\newtheorem{theorem}{Theorem}
\newtheorem{definition}[theorem]{Definition}
\newtheorem{lemma}[theorem]{Lemma}
\newtheorem{corollary}[theorem]{Corollary}

\newcounter{assp}
\newtheorem{assp}[assp]{Assumption}

\theoremstyle{myremarkstyle}
\newtheorem{remark}{Remark}

\theoremstyle{myproofstyle}
\newtheorem*{PROOF}{Proof}

\usepackage{fancyvrb} % Verbatim
\DefineVerbatimEnvironment{verbatim}{Verbatim}{xleftmargin=.3in,fontsize=\scriptsize,formatcom=\color{black}}

\DeclareMathOperator{\Frac}{Frac}
\newcommand{\JB}{{J_{\bullet}}}
\newcommand{\KB}{{K_{\bullet}}}

\newcommand{\down}{\tc{hcolor}{\alpha}}
\newcommand{\up}{\tc{hcolor}{\beta}}
\newcommand{\pbc}{\tc{hcolor}{\gamma}}
\newcommand{\polyann}{\tc{scolor}{\text{poly}}}
\newcommand{\modp}{\mod \polyann}
\newcommand{\Dw}{\widetilde{D}}

\title{\boldmath Scattering amplitude annihilators}

\author[a]{Andrea N\"utzi}
\author[b]{and Michael Reiterer}

\affiliation[a]{ETH Zurich, Switzerland}
\affiliation[b]{The Hebrew University of Jerusalem, Israel}

\emailAdd{andrea.nuetzi@math.ethz.ch}
\emailAdd{michael.reiterer@protonmail.com}

\abstract{Several second order differential operators are shown
to annihilate the YM and GR tree scattering amplitudes.
In particular we prove a conjecture of Loebbert, Mojaza and Plefka
from their investigation of a hidden conformal symmetry in GR.
}

\begin{document} 
\maketitle
\flushbottom

\section{Introduction}
The scattering amplitudes of Yang-Mills (YM) and general relativity (GR)
are interesting objects in physics.
The dimension-neutral tree amplitudes are rational functions of many variables;
see Appendix \ref{app:formulas} for examples.
As with other special functions,
it is natural to ask what differential identities they satisfy.
We show that a number of second order differential
operators annihilate the dimension-neutral YM and GR amplitudes.
\step
This paper is motivated by recent work of
Loebbert, Mojaza and Plefka \cite{lmp}.
They investigated a potential hidden conformal symmetry of GR tree amplitudes
in general spacetime dimension $d$. 
This led them to conjecture certain identities
for the dimension-neutral amplitudes
that they verified by explicit calculation for
amplitudes with $n = 3,4,5,6$ legs\footnote{%
We thank 
F.~Loebbert, M.~Mojaza and J.~Plefka for
making available computer code to check the $n=5$ case
and for answering questions we had about their paper \cite{lmp}.
}.
They conjectured that they would continue to hold for $n \geq 7$.
This paper contains:
\begin{itemize}
\item A
reformulation of these identities as proper differential identities for the amplitudes.
The formulation in \cite{lmp} is more intricate,
as discussed in Appendix \ref{app:translation}.
\item A proof of these identities for all $n \geq 3$, by induction on $n$.
\end{itemize}
We do not take up the interesting theme of hidden conformal symmetry \cite{lmp}.
The identities conjectured in \cite{lmp}
and proved here do not readily imply identities for the amplitudes
in specific dimensions such as $d=4$,
an interesting question that we leave open.
\step
In this paper we exclusively work with the dimension-neutral version of the tree amplitudes;
the spacetime dimension $d$ will be absent. These amplitudes
are rational functions on a complex vector space of dimension $2n(n-2)$
with simple poles along a constellation of hyperplanes.
The coordinates on this vector space will be denoted
$k_{ij}$, $c_{ij}$, $e_{ij}$
with subscripts running over an index set\footnote{%
The YM amplitude requires a cyclic order on $I$,
whereas the GR amplitude is permutation invariant and $I$ is an unordered set.
We will not mention this further in this introduction.}
$I$ with $|I| = n$.
The following linear relations among the coordinates cut the dimension of this vector space
down to $2n(n-2)$:
\begin{equation}\label{eq:linrel}
\begin{aligned}
k_{ii} & = 0 &\qquad k_{ij} - k_{ji} & = 0 &\qquad \textstyle\sum_{i \in I} k_{ij} & = 0\\
c_{ii} & = 0 && & \textstyle\sum_{i\in I} c_{ij} & = 0 \\
e_{ii} & = 0 & e_{ij} - e_{ji} & = 0
\end{aligned}
\end{equation}
The amplitudes are actually polynomial in the $c_{ij}$ and $e_{ij}$ variables.
To obtain the amplitudes in $d$ spacetime dimensions set
$k_{ij} = k_i \cdot k_j$ and $c_{ij} = k_i \cdot \eps_j$
and $e_{ij} = (1-\delta_{ij})\eps_i \cdot \eps_j$
where $k_i$ and $\eps_i$ are $d$-dimensional vectors,
the momentum and polarization of the $i$-th particle.
\vskip 2mm
We need a workable definition of the amplitudes for general $n$.
For the purpose of this paper, the amplitudes are defined by the recursion in Definition \ref{def:rec}.
This recursion is based on the factorization of residues
formulated directly in $k_{ij}$, $c_{ij}$, $e_{ij}$ variables\footnote{%
The original recursion of this kind is BCFW recursion \cite{bcfw}.
The recursion we use in this paper does not involve BCFW shifts
nor decay conditions under such shifts,
it is only required that the amplitudes factor properly.
The proof that this determines the amplitudes uniquely is along the lines of \cite{nr},
which is for $d=4$, but is simpler in the present dimension-neutral case since
one is on a vector space and all poles are along linear subspaces.
It is useful in this paper to have the uniqueness 
statement in a simple form since the same
argument (cf.~Lemmas \ref{lemma:uniq} and \ref{lemma:vanish})
also plays a key role in proof of the main theorem, Theorem \ref{theorem:main}.}.
On the face of it, it is neither obvious that this recursion
admits a solution, nor that the solution is unique.
A sketch of existence using Feynman rules is in Remark \ref{remark:existence},
but since we are unaware of a reference that spells
this out in detail for the dimension-neutral amplitudes,
we logically state existence as Assumption \ref{asspa} below.
Uniqueness is shown in Lemma \ref{lemma:uniq}.
\begin{assp}\label{asspa}
A sequence of rational functions
satisfying the recursion in Definition \ref{def:rec} exists.
We refer to them as the YM respectively GR amplitudes.
\end{assp}
In stating our main result below,
we avoid standard notation for partial derivatives
such as $\p_{k_{ij}}$, $\p_{c_{ij}}$, $\p_{e_{ij}}$
because it is ambiguous how these act on functions
on the vector space defined by \eqref{eq:linrel}.
To illustrate, if $I = \{1,\ldots,5\}$
then $c_{12}$ and $-c_{32}-c_{42}-c_{52}$
are two representatives of one the same function,
but if we act on them with $\p_{c_{12}}$ using a naive interpretation of
the partial derivative, then we obtain conflicting results, $1$ and $0$ respectively.
To fix this
we introduce new notation $D_{k_{ij}}$, $D_{c_{ij}}$, $D_{e_{ij}}$
which are differential operators
that act unambiguously on functions on the vector space defined by \eqref{eq:linrel}.
To illustrate, if $I = \{1,\ldots,5\}$ 
then we use
$D_{c_{12}} = \p_{c_{12}} - \tfrac14(\p_{c_{12}} + \p_{c_{32}} + \p_{c_{42}} + \p_{c_{52}})$
which acts unambiguously.
The $D_{k_{ij}}$ are linear combinations of the $\p_{k_{ij}}$
with constant coefficients;
the $D_{c_{ij}}$ are linear combinations of the $\p_{c_{ij}}$
with constant coefficients;
the $D_{e_{ij}}$ are linear combinations of the $\p_{e_{ij}}$
with constant coefficients.
Their complete definition is in \mbox{Lemma \ref{lemma:dops}},
and we have also included computer code in Appendix \ref{app:code}.
We note that our definition of the $D$-operators is
also a matter of convention;
we have merely made a convenient choice
of operators that span the $2n(n-2)$-dimensional
space of constant first order differential operators along
the vector space defined by \eqref{eq:linrel}.
\step
There are $2n+1$ first order differential operators that
are well-known to annihilate the amplitudes:
$n$ of them express gauge invariance;
$n$ express homogeneity in the polarizations;
and one expresses homogeneity jointly in all momenta.
They are, respectively,
\begin{subequations}\label{eq:xyz}
\begin{align}
X_i & = \textstyle\sum_{j \in I}(
k_{ji}D_{c_{ji}} + c_{ij}D_{e_{ij}})\\
Y_i & = \textstyle\sum_{j\in I}(c_{ji} D_{c_{ji}} + e_{ij} D_{e_{ij}}) - h\\
Z & = \textstyle\sum_{i,j\in I}(k_{ij} D_{k_{ij}} + c_{ij} D_{c_{ij}}) - s
\end{align}
\end{subequations}
where $h = 1$ and $s = 4-n$ for YM respectively $h = 2$ and $s = 2$ for GR.
Polynomiality in the $c_{ij}$ and $e_{ij}$ variables
implies further obvious annihilators.
There are, in addition,
several second order annihilators as we show in this paper.
%---------------------------------------------------
%---------------------------------------------------
\begin{theorem}\label{theorem:main}
Define the YM and GR tree amplitude by the recursion in Definition \ref{def:rec}
and make Assumption \ref{asspa}.
The differential operators $A_i$, $B_i$, $C_i$
defined below annihilate the tree amplitudes
for all $n = |I| \geq 3$ and $i \in I$:
\begin{subequations}\label{eqs:abc}
\begin{equation}\label{id1}
A_i = \textstyle\sum_{j,k \in I}
\big(
  \tfrac12 k_{jk} D_{c_{ji}} D_{c_{ki}}
+ c_{jk} D_{c_{ji}} D_{e_{ki}}
+ \tfrac12 e_{jk} D_{e_{ji}} D_{e_{ki}}
\big)
\end{equation}
and
\begin{multline} \label{eqs:b}
B_i =
\textstyle\sum_{j,k \in I}\big(
  c_{jk} D_{k_{ij}} D_{e_{ik}}
+ e_{jk} D_{c_{ij}} D_{e_{ik}}
+ k_{jk} D_{k_{ij}} D_{c_{ki}}
+ c_{kj} D_{c_{ij}} D_{c_{ki}}\\
- c_{kj} D_{k_{jk}} D_{e_{ji}}
- e_{jk} D_{c_{jk}} D_{e_{ji}}
- k_{jk} D_{k_{jk}} D_{c_{ji}}
- c_{jk} D_{c_{ji}} D_{c_{jk}}
\big)
\end{multline}
and
$C_i = \widetilde{C}_i - \textstyle\frac{1}{n} \sum_{j \in I} \widetilde{C}_j$
where
\begin{multline} \label{eqs:c}
\widetilde{C}_i =
\textstyle\sum_{j,k \in I}\big(
\tfrac12 e_{jk} D_{c_{ij}} D_{c_{ik}}
+ c_{kj} D_{k_{ik}} D_{c_{ij}}
+ \tfrac12 k_{jk} D_{k_{ij}} D_{k_{ik}}\\
- e_{jk} D_{c_{ij}} D_{c_{jk}}
- c_{kj} D_{k_{jk}} D_{c_{ij}} 
- c_{jk} D_{k_{ji}} D_{c_{jk}}
- k_{jk} D_{k_{ji}} D_{k_{jk}}
\big)
\end{multline}
\end{subequations}
The $D$-operators are defined in Lemma \ref{lemma:dops}.
\end{theorem}
%----------------------------------------------------
That $B_i$, $C_i$ 
annihilate the GR amplitude 
was conjectured in \cite{lmp},
using a more intricate formulation
that we discuss in Appendix \ref{app:translation}.
The formulation of these identities
directly as annihilating differential operators
in Theorem \ref{theorem:main} is therefore new\footnote{%
It allows us, for instance, to compute a number of commutators
in Section \ref{sec:comms}.
They imply that every function annihilated by $X_i$, $Y_i$, $Z$, $C_i$
is automatically annihilated by $A_i$, $B_i$.}.
Our proof of Theorem \ref{theorem:main} is by induction on $n$
and is summarized later in this introduction.

\begin{remark} \label{remark:interpretation}
The amplitudes in $d \geq 4$ spacetime dimensions are obtained
from the dimension-neutral ones by setting
$k_{ij} = k_i \cdot k_j$ and $c_{ij} = k_i \cdot \eps_j$
and $e_{ij} = (1-\delta_{ij})\eps_i \cdot \eps_j$
where $k_i$ and $\eps_i$
are elements of a $d$-dimensional complex vector space,
and $\cdot$ is a nondegenerate symmetric bilinear pairing\footnote{%
For instance, $\C^d$ with standard pairing
$z\cdot w = z_1w_1 + \ldots + z_d w_d$.
The choice does not matter since all such vector spaces with nondegenerate pairing
are isomorphic.
}.
The momentum vectors $k_i$ are
subject to $k_i \cdot k_i = 0$ and momentum conservation $\sum_i k_i = 0$,
and the polarization vectors $\eps_i$ are subject to $k_i \cdot \eps_i = 0$.
Assume $k_i \neq 0$ here, so $\eps_i$ lies in a subspace of dimension $d-1$.
For every $i$, the YM amplitude is linear in $\eps_i$,
the GR amplitude is quadratic in $\eps_i$; this is witnessed by the annihilator $Y_i$.
For every $k \neq 0$ with $k \cdot k = 0$ abbreviate
$P(k) = \{ \eps \mid k \cdot \eps=0\}/\C k$
which is a vector space of dimension $d-2$,
and observe that $\cdot$ induces a nondegenerate symmetric bilinear pairing on $P(k)$.
Then, separately for every $i$ and fixed $k_i \neq 0$:
\begin{itemize}
\item As a function of $\eps_i$,
the YM amplitude descends to a linear form on $P(k_i)$.
\item As a function of $\eps_i$,
the GR amplitude descends to a quadratic form on $P(k_i)$\footnote{%
Equivalently, a linear form on the symmetric tensor product $S^2 P(k_i)$.
One can decompose this into the trace and the traceless part relative
to the symmetric bilinear pairing on $P(k_i)$.}.
\end{itemize}
This is witnessed by $X_i$ and is known as gauge invariance.
The amplitudes are also homogeneous jointly in all momenta,
as witnessed by $Z$.
In summary, the $d$-dimensional amplitudes are sections of certain vector bundles
on the projective
variety $k_i \cdot k_i = 0$ and $\sum_i k_i = 0$\footnote{%
Actually one only has a vector bundle away from the singular locus.
}\textsuperscript{,}\footnote{%
It would be interesting to see whether the annihilators
$A_i$, $B_i$, $C_i$ or suitable linear combinations of them imply annihilators
for the $d$-dimensional amplitudes.
The latter would be differential operators on the vector bundle
on which the amplitudes live,
possibly taking values in another vector bundle.
}\textsuperscript{,}\footnote{%
The fact that $A_i$ annihilates is vacuous for YM,
since $A_i$ involves two derivatives
with respect to the polarization of the $i$-th particle.
It is not vacuous for GR amplitudes,
in the sense that one can write down a rational function
annihilated by the obvious annihilators listed before Theorem \ref{theorem:main}
but not by $A_i$.
}.
In this paper we do not work with momenta $k_i$ and polarizations $\eps_i$.
Our primary variables are $k_{ij}$, $c_{ij}$, $e_{ij}$ and we exploit the fact that
the relations \eqref{eq:linrel} are linear.
\end{remark}
%%%%%%%%%%%%%%%%%%%%%%%%%%%%%%%%%%%%%%%%%%%%%%%%%%%%%%%%%%%%%%%%%%%
%%%%%%%%%%%%%%%%%%%%%%%%%%%%%%%%%%%%%%%%%%%%%%%%%%%%%%%%%%%%%%%%%%%

Let us outline the proof of Theorem \ref{theorem:main},
which is by brute force by induction on $n$.
It uses the recursive definition of the amplitudes,
by which the amplitudes have simple poles along a constellation of hyperplanes,
each residue being equal to a product of two lower amplitudes.
Denoting by $M_I$ the GR amplitude for index set $I$, 
schematically one has
\begin{equation}\label{eq:rmx}
\Res M_I \;=\; U^2 M_{J \sqcup \{\bullet\}} M_{K \sqcup \{\bullet\}}
\end{equation}
for every decomposition $I = J \sqcup K$ into two subsets of two or more elements,
$|J| \geq 2$ and $|K| \geq 2$.
The residue is taken along the locus where $k_J = \sum_{j \in J} k_j$
satisfies $k_J \cdot k_J = 0$.
Translating to the coordinates we use, this locus is the hyperplane
\begin{equation}\label{eq:kjj}
    \xi\;\stackrel{\textnormal{def}}{=}\; \textstyle\sum_{j,j' \in J} k_{jj'} = 0
\end{equation}
On the right hand side of \eqref{eq:rmx},
the bullet $\bullet$ stands for one more particle.
Properly defining the right hand side of \eqref{eq:rmx}
requires contracting the polarizations
of the respective $\bullet$ particles in $M_{J \sqcup \{\bullet\}}$
and $M_{K \sqcup \{\bullet\}}$.
This contraction is effected by $U^2$,
using the differential operator $U$ in Definition \ref{def:u}.
Suppose now we are at the $n$-th induction step.
Within each induction step, one first shows that the $A_i$ annihilate,
then the $B_i$, then the $C_i$.
The structure of the argument is analogous in each case,
so suppose for concreteness that we want to show $C_i M_I = 0$ for some $i$.
This is done in two steps:
\begin{enumerate}
\item Show that $C_i M_I$ has no pole. Namely,
for every decomposition $I = J \sqcup K$ as above one must check
that $C_i M_I$ has no pole along the hyperplane $\xi = 0$ in \eqref{eq:kjj}.
\item Show that the complete absence of poles in $C_iM_I$,
together with other known properties of $C_i M_I$ such as homogeneity
properties, imply $C_i M_I = 0$.
\end{enumerate}
In Step 1, note that $M_I$ has a simple $\xi^{-1}$ type pole,
so $C_i$ being second order, the 
application $C_i M_I$ can naively have a third order pole,
terms of type $\xi^{-3}$ and $\xi^{-2}$ and $\xi^{-1}$. 
While the $\xi^{-3}$ term is easily seen to be absent,
the absence of $\xi^{-2}$ and $\xi^{-1}$
terms is by a lengthy explicit calculation that exploits
only the annihilators of $M_{J \sqcup \{\bullet\}}$ and $M_{K \sqcup \{\bullet\}}$
that are known by the induction hypothesis.
In a nutshell, and using schematic notation again, in Step 1 one must check that
$C_i(\tfrac{1}{\xi} U^2 M_{J \sqcup \{\bullet\}} M_{K \sqcup \{\bullet\}})$
has no pole.
The computations establishing Step 1
are in the most technical lemma of this paper, Lemma \ref{lemma:key}.
Step 2 relies on Lemma \ref{lemma:vanish}
(which is the same lemma that also establishes that the recursion
determines the amplitudes uniquely)
and on the commutators in Lemma \ref{lemma:comms}.
The YM case is entirely analogous.

The proof does not use explicit formulas for the amplitudes.
Instead, the recursion for the amplitudes using \eqref{eq:rmx} is effectively turned
into a recursion for the annihilators.
\step
One can ask what the full annihilator of the YM respectively GR tree amplitudes is.
To every rational function one can associate its annihilator,
a left ideal in the Weyl algebra of polynomial differential operators.
Rational functions are holonomic,
meaning their annihilator is 
as big as allowed by the Bernstein inequality\footnote{%
This says that over $\R$ or $\C$,
and for a proper left ideal $I \subset D$ in the Weyl algebra $D$
in $x_1,\ldots,x_m,\p_1,\ldots,\p_m$,
the dimension of $D/I$
(defined to be the degree of a suitable Hilbert polynomial) is $\geq m$.
Note that a rational function $g/f$ always has 
$m$ first order annihilators of the form
$f (\p_i g) - g (\p_i f) - gf \p_i \in D$ 
but they are in general of high polynomial degree.
See also \cite{z}.}.
In a Weyl algebra, every left ideal is finitely generated\footnote{%
Somewhat surprisingly, by a theorem of Stafford,
every left ideal in a Weyl algebra over $\R$ or $\C$
can in fact be generated by just two elements.}.
Algorithms are available to determine the annihilator of a rational function,
e.g.~in Macaulay2 \cite{m2,jl},
and while in principle such tools can be applied to tree amplitudes,
this approach does not seem to be practically viable for general $n$.
\step
It would be interesting to understand if there is a
more conceptual explanation for these annihilators
perhaps building on ideas in \cite{lmp};
to inquire if there are corresponding identities for
the amplitudes in specific dimensions such as $d=4$;
and to try to apply these annihilators to derive properties of the amplitudes.

%------------------------------------------------------------------
\newcommand{\hd}[1]{{\bf #1}}
\section{Preliminaries}

\hd{Notation.} All vector spaces and algebras are over the complex numbers.
For every finite-dimensional vector space $X$ we denote:
\begin{center}
\begin{tabular}{r@{\qquad}l}
$X^\ast$ & the dual vector space\\
$R_X$ & the commutative algebra of polynomials $X \to \C$\\
$D_X$ & the Weyl algebra on $X$\\
$\Frac(R_X)$ & the field of rational functions on $X$
\end{tabular}
\end{center}
These things are defined independently of coordinates.
This is important because the spaces we encounter do not have a canonical coordinate system,
and we work with a linearly dependent set of coordinates.
Canonically,
\begin{center}
\begin{tabular}{r@{\qquad}l}
$X^\ast \subset R_X$ & the linear functionals, we refer to them as coordinates\\
$X \subset D_X$ & first order constant coefficient differential operators
\end{tabular}
\end{center}
As vector spaces,
$R_X \simeq S X^\ast$ and $D_X \simeq S X^\ast \otimes S X$ where
$S$ is the symmetric tensor algebra. 
Here $SX$ are the constant coefficient differential operators.
\step
\hd{Coordinate dependent definitions.}
Even though we do not commit to a particular basis,
we recall the coordinate dependent definitions.
By a coordinate we mean an element of $X^\ast$.
A coordinate system is a basis $x_1,\ldots,x_m \in X^\ast$.
The polynomial ring is then
\[
R_X = \C[x_1,\ldots,x_m]
\]
Let $\p_1,\ldots,\p_m \in X$ be the dual basis.
Then $D_X$ is the associative algebra with identity generated by 
the variables $x_1,\ldots, x_m, \p_1,\ldots, \p_m$ modulo the two-sided ideal generated by
\[
 x_ix_j-x_jx_i \qquad \p_i\p_j - \p_j\p_i\qquad \p_ix_j-x_j\p_i - \delta_{ij}
\]
\step
\hd{Linear maps.}
If $X,Y$ are vector spaces then
every linear map $\down: X \to Y$ canonically induces several other linear maps:
\begin{center}
\begin{tabular}{r@{\qquad}l}
$SX \to SY$ & the push-forward, also denoted $\down$\\
$Y^\ast \to X^\ast$ & the adjoint map, denoted $\down^\ast$\\
$R_Y \to R_X$ & the pull-back, also denoted $\down^\ast$\\
$\Frac(R_Y) \to \Frac(R_X)$ & the pull-back, also denoted $\down^\ast$
\end{tabular}
\end{center}
It neither induces a map $D_X \to D_Y$ nor $D_Y \to D_X$.
We often find it convenient to specify a linear map by specifying 
the adjoint $Y^\ast \to X^\ast$.
\step
\hd{Direct sums.}
For a direct sum of finite-dimensional vector spaces,
there are canonical isomorphisms
$R_{X \oplus Y} \simeq R_X \otimes R_Y$
and $D_{X \oplus Y} \simeq D_X \otimes D_Y$
that we frequently use.
This uses
the tensor product of algebras,
where all elements of $D_X$ commute with all elements of $D_Y$.
Canonical inclusions such as
$D_X \hookrightarrow D_X \otimes D_Y, \delta \mapsto \delta \otimes 1$ are sometimes used.
\step
\hd{Index sets.}
Instead of a standard index set such as $\{1,\ldots,n\}$
we work with finite index sets denoted $I$, hence $n$ is replaced by $|I|$.
For YM, a cyclic ordering of the elements of $I$ is required.
In many calculations, $I$ is a disjoint union
\begin{subequations}\label{eq:ijk}
\begin{equation}\label{eq:ijk0}
I = J \sqcup K
\end{equation}
with $|J|,|K| \geq 2$
and often \eqref{eq:ijk0} is assumed.
Sometimes we need index sets
with a distinguished element, always denoted $\bullet$.
We abbreviate
\begin{equation}
\JB = J \sqcup \{\bullet\}
\qquad \KB = K \sqcup \{\bullet\}
\end{equation}
\end{subequations}
If $\JB$, $\KB$ have a cyclic order
then $J \sqcup K$ acquires a cyclic order\footnote{%
If the elements of $\JB$ are ${\bullet},a_1,\ldots,a_{|J|}$ and
the elements of $\KB$ are ${\bullet},b_1,\ldots,b_{|K|}$ listed in cyclic order,
starting with ${\bullet}$ for convenience,
then on $J \sqcup K$ use the cyclic order $a_1,\ldots,a_{|J|},b_1,\ldots,b_{|K|}$.}.
Conversely, if $I$ has a cyclic order and \eqref{eq:ijk0} respects this,
then $\JB$, $\KB$ inherit a cyclic order.
%--------------------------------------------------------------------------
\section{Kinematic variables} \label{sec:kinematic}
Here we define in detail the vector space on which the amplitudes live;
the differential operators $D$ used in Theorem \ref{theorem:main};
and notation that allows one to state the schematic factorization 
of residues in \eqref{eq:rmx} in a precise way in Definition \ref{def:rec}.
\begin{definition}\label{def:cvecs}
For every finite set $I$ with $|I| \geq 3$,
consider first the complex `ambient' vector space of dimension $3|I|^2$ with
coordinate system
\begin{equation}\label{eq:kce}
k_{ij}, c_{ij}, e_{ij} \qquad \text{where} \qquad i,j \in I
\end{equation}
Denote by $C(I)$ the linear subspace defined by the relations \eqref{eq:linrel}.
In this paper, the elements \eqref{eq:kce} are understood to be elements of the
dual space $C(I)^\ast$.
\end{definition}
The dimension of $C(I)$ is
\[
\underbrace{\tfrac12 n(n-3)}_{k} + \underbrace{n(n-2)}_{c} + \underbrace{\tfrac12 n(n-1)}_{e}
= 2n(n-2)
\]
where $n = |I|$.
If $|I|=3$ then the $k_{ij}$ vanish identically as elements of $C(I)^\ast$.
Note that the relations \eqref{eq:linrel} do not refer to an ordering,
hence there is a natural action of the group of permutations of $I$ on the vector space $C(I)$.
%---------------------------------------------------------------
\step
The \eqref{eq:kce} are a linearly dependent set of coordinates,
not a coordinate system on $C(I)$. Therefore we cannot define partial derivatives in the usual way.
We work with the following linearly dependent
set of constant coefficient differential operators.
\begin{lemma} \label{lemma:dops}
There are unique
\begin{equation}\label{eq:dkce}
D_{k_{ij}}, D_{c_{ij}}, D_{e_{ij}} \in C(I)
\end{equation}
that as elements of the Weyl algebra $D_{C(I)}$ satisfy
\begin{align*}
[D_{k_{ij}}, k_{ab}] & = (1-\delta_{ij})(1-\delta_{ab})(
\delta_{ia}\delta_{jb} + \delta_{ib}\delta_{ja}
- \tfrac{1}{|I|-2} (\delta_{ia} + \delta_{ib} + \delta_{ja} + \delta_{jb})
+ \tfrac{2}{(|I|-1)(|I|-2)}
)\\
[D_{c_{ij}}, c_{ab}] & = (1-\delta_{ij})(1-\delta_{ab})(\delta_{ia}\delta_{jb}- \tfrac{1}{|I|-1} \delta_{jb})\\
[D_{e_{ij}}, e_{ab}] & = (1-\delta_{ij})(1-\delta_{ab})(\delta_{ia}\delta_{jb} + \delta_{ib}\delta_{ja})
\end{align*}
and such that all `mixed' commutators are zero:
\[
\begin{aligned}
[D_{k_{ij}},c_{ab}] & = 0\quad &
[D_{c_{ij}},k_{ab}] & = 0\quad &
[D_{e_{ij}},k_{ab}] & = 0 \\
[D_{k_{ij}},e_{ab}] & = 0 &
[D_{c_{ij}},e_{ab}] & = 0 &
[D_{e_{ij}},c_{ab}] & = 0
\end{aligned}
\]
They span $C(I)$.
They satisfy
\begin{equation}\label{eq:dlinrel}
\begin{aligned}
D_{k_{ii}} & = 0\quad & D_{k_{ij}} & = D_{k_{ji}}\quad & \textstyle\sum_i D_{k_{ij}} & = 0\\
D_{c_{ii}} & = 0 &              &                & \textstyle\sum_i D_{c_{ij}} & = 0\\
D_{e_{ii}} & = 0 & D_{e_{ij}} & = D_{e_{ji}} &              &
\end{aligned}
\end{equation}
\end{lemma}
\begin{proof}
The given commutators for $D_{k_{ij}}$
at first only determine an operator on the ambient vector space.
But since $[D_{k_{ij}},-]$ annihilates all
left hand sides of the relations \eqref{eq:linrel},
we obtain a unique constant coefficient differential
operator $\in C(I)$ as claimed.
The rest of the lemma goes the same way.
\qed\end{proof}
%--------------------
We have introduced coordinates $k_{ii}$, $c_{ii}$, $e_{ii}$
and derivatives $D_{k_{ii}}$, $D_{c_{ii}}$, $D_{e_{ii}}$
that are identically zero and in some sense superfluous.
These phantom objects are nevertheless useful because they allow us
to write sums as in \eqref{eqs:abc}.

%---------------------------------------------------------------

\newcommand{\kp}{k^\perp}
\newcommand{\cp}{c^\perp}
\newcommand{\ep}{e^\perp}
\newcommand{\Dp}{D^\perp}
\newcommand{\Mp}{M^\perp}
\newcommand{\Cp}{C(I)^\perp}

\step
Convention \eqref{eq:ijk} is in force, $I = J \sqcup K$.
To this decomposition of $I$ we associate a coordinate $\xi$,
and a vector $\Xi$ transversal to the hyperplane $\xi = 0$.
The dependence of these definitions on the decomposition $I = J \sqcup K$
is left implicit in the notation.
\begin{figure}[h!]
\begin{center}
  \begin{tikzpicture}[scale=1.2]
  \draw [black] (0,0)
  to (5,0.5)
  to (5.5,1.5) node [anchor=west,yshift=-10] {$\mathrlap{\Cp = \{\xi=0\}}$}
  to (0.5,1)
  to cycle;
  \draw [black,thick,->] (2.55,0.80) -- (2.35,2.3) node [anchor=west,xshift=3] {$\Xi$};
  \draw [black,thick,->] (2.55,0.80) -- (3.80,0.91) node [anchor=west] {$\Dp$};
\end{tikzpicture}
\end{center}
\caption{%
Decomposition of $C(I)$ associated to the decomposition $I = J \sqcup K$.
The coordinates $\kp$, $\cp$, $\ep$ are uniquely determined by
being constant in the direction $\Xi$
and coinciding with $k$, $c$, $e$ respectively along $\Cp$.
The derivatives $\Dp$ span $\Cp$.
}\label{fig:mp}
\end{figure}
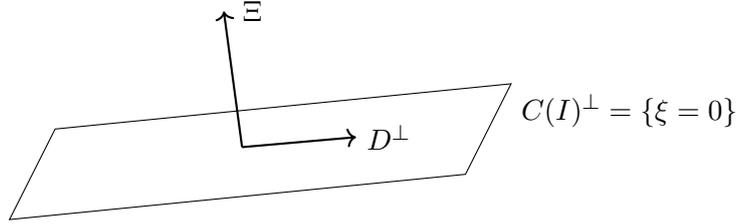
%%%%%%%%%%%%%%%%%%%%%
\begin{definition}\label{def:xixi}
Define $\xi \in C(I)^\ast$ and $\Xi \in C(I)$ by\footnote{%
One can equivalently define $\Xi$ by
$\Xi = \text{const}\cdot \sum_{j,j' \in J} D_{k_{jj'}}$
normalized so that $[\Xi,\xi]=2$.}
\begin{align}
\label{eq:defxi}
\xi & = \textstyle\sum_{j,j' \in J} k_{jj'}
    = \textstyle\sum_{k,k' \in K} k_{kk'}\\
\Xi & = \textstyle\frac{1}{|J|(|J|-1)} \sum_{j,j' \in J} D_{k_{jj'}}
 \textstyle- \frac{2}{|J||K|} \sum_{j \in J, k \in K} D_{k_{jk}}
+ \frac{1}{|K|(|K|-1)} \sum_{k,k' \in K} D_{k_{kk'}}
\end{align}
Also define
\begin{align*}
\kp_{jj'} & = k_{jj'} - \textstyle\frac{1-\delta_{jj'}}{|J|(|J|-1)} \xi &
\Dp_{k_{jj'}} & = D_{k_{jj'}} - \textstyle\frac{|K|(|K|-1)(1-\delta_{jj'})}{(|I|-1)(|I|-2)} \Xi \\
%-----
\kp_{jk} = \kp_{kj} & = k_{jk} + \textstyle\frac{1}{|J||K|} \xi &
\rule{10pt}{0pt}
\Dp_{k_{jk}} = \Dp_{k_{kj}} & = D_{k_{jk}} + \textstyle\frac{(|K|-1)(|J|-1)}{(|I|-1)(|I|-2)} \Xi  \displaybreak[0] \\
%------
\kp_{kk'} & = k_{kk'} - \textstyle\frac{1-\delta_{kk'}}{|K|(|K|-1)} \xi &
\Dp_{k_{kk'}} & = D_{k_{kk'}} - \textstyle\frac{|J|(|J|-1)(1-\delta_{kk'})}{(|I|-1)(|I|-2)} \Xi  \displaybreak[0] \\
%------
\cp_{ii'} & = c_{ii'} &
\Dp_{c_{ii'}} & = D_{c_{ii'}} \\
%------
\ep_{ii'} & = e_{ii'} &
\Dp_{e_{ii'}} & = D_{e_{ii'}}
\end{align*}
for all $j,j' \in J$ and $k,k' \in K$ and $i,i' \in I = J \sqcup K$.
Set $\Cp = \{\xi = 0\}$.
\end{definition}
%------------------------------------------------
\begin{lemma}\label{lemma:dpx} We have
\[
   C(I) = \C \Xi \oplus \Cp
\]
The $\Dp$ span $\Cp$.
 In the Weyl algebra $D_{C(I)}$, 
\begin{align*}
[\Xi,\xi] & = 2 &
[\Xi,k_{jj'}] & = \textstyle\frac{2(1-\delta_{jj'})}{|J|(|J|-1)}\\
[\Xi,\kp] &= [\Xi,\cp] = [\Xi,\ep] = 0 &
[\Xi,k_{jk}] & = \textstyle-\frac{2}{|J||K|}\\
[\Dp, \xi] &= 0 &
[\Xi,k_{kk'}] & = \textstyle\frac{2(1-\delta_{kk'})}{|K|(|K|-1)}
\end{align*}
for all $j,j' \in J$ and $k,k' \in K$.
The $\kp, \cp, \ep$ and $\Dp$ satisfy
the linear relations \eqref{eq:linrel} 
and \eqref{eq:dlinrel}, with $k$ replaced by $\kp$ and so forth.
They additionally satisfy
\begin{equation}\label{eq:kperp}
\textstyle\sum_{j,j' \in J} \kp_{jj'} = 0
\qquad
\textstyle\sum_{j,j' \in J} \Dp_{k_{jj'}} = 0
\end{equation}
Same if the summation is over $k,k' \in K$ instead.
\end{lemma}
%---------------------------------------------------------
\begin{proof}
By direct calculation.
\qed\end{proof}
%----------------------------------------------------------
Let us pause to discuss what is needed
to state the factorization \eqref{eq:rmx} precisely.
Essentially we need maps between
the $\xi=0$ hyperplane in $C(I)$ which is the locus
of the pole under consideration on the one hand,
and the spaces $C(\JB)$ and $C(\KB)$
where the two amplitudes live on the other hand.
Here are more details:
\begin{itemize}
\item Since $C(I)$ lacks variables corresponding
to the polarization of the $\bullet$ particle,
we use an extension $M = C(I) \oplus E$ where $E$
provides the missing variables for $\bullet$.
\item We define maps not only along $\xi = 0$
(now understood to be a hyperplane in $M$) but on all of $M$.
This leads to two linear, surjective maps
\begin{equation}\label{mdiag}
\begin{tikzcd}[column sep = huge]
& M \arrow[dl, swap,"\down_J"]
    \arrow[dr, "\down_K"] & \\
C(\JB)
& &
C(\KB)
\end{tikzcd}
\end{equation}
The extension to $\xi \neq 0$ is not needed to state the recursion
in Definition \ref{def:rec}, but it is exploited in the statement and proof
of the technical Lemma \ref{lemma:key}.
\item Along $\xi = 0$ the maps $\down_J$ and $\down_K$
are defined in the canonical way,
by identifying the momenta $k_i$ and polarizations $\eps_i$
with the corresponding ones in $\JB$ when $i\in J$
respectively $\KB$ when $i \in K$,
and by translating this to the variables \eqref{eq:kce},
cf.~Remark \ref{remark:interpretation}.
Together with the linear constraints, this defines the maps completely along $\xi = 0$.
The extension to $\xi \neq 0$
is such that $\down_J$ and $\down_K$ are constant along $\Xi$.
\item Finally, a differential operator $U$ is introduced on $M$ whose
purpose is to contract the polarization of the particle $\bullet$ in $\JB$
with the corresponding particle $\bullet$ in $\KB$.
\end{itemize}
In the remainder of this section, we spell these things out.
%%%%%%%%%%%%%%%%%%%%%%%%%%%%%%%%%%%%%%%%%%%%%%%%%%%%%%%%%%%%%%%%%%%%%%%%%%%%%%%
\newcommand{\ph}{\widehat{\p}}
\begin{definition}[Extension space]
Consider first the complex `ambient' vector space of dimension $2|J|+2|K| = 2|I|$
with coordinate system\footnote{%
In this instance, the bullet $\bullet$ is nothing but suggestive decoration.}
\begin{equation}\label{eq:evars}
c_{j\bullet},\, e_{j\bullet},\, c_{k\bullet},\, e_{k\bullet}
\qquad \text{where}
\qquad j \in J, k\in K
\end{equation}
Let $E$ be the subspace defined by
\[
\textstyle\sum_{j \in J} c_{j\bullet} = 0
\qquad
\textstyle\sum_{k \in K} c_{k\bullet} = 0
\]
a linear subspace of dimension $2|J| + 2|K|-2 = 2|I|-2$.
The elements \eqref{eq:evars} are understood to be elements of the dual space $E^\ast$.
There are unique
\[
D_{c_{j\bullet}},\, D_{e_{j\bullet}},\, D_{c_{k\bullet}},\, D_{e_{k\bullet}}\; \in E
\]
that as elements of the Weyl algebra $D_E$ satisfy
\begin{align*}
[D_{c_{j\bullet}}, c_{j'\bullet}] & = \delta_{jj'} - \tfrac{1}{|J|} &
[D_{c_{k\bullet}}, c_{k'\bullet}] & = \delta_{kk'} - \tfrac{1}{|K|}\\
[D_{e_{j\bullet}}, e_{j'\bullet}] & = \delta_{jj'} &
[D_{e_{k\bullet}}, e_{k'\bullet}] & = \delta_{kk'}
\end{align*}
for $j,j' \in J$ and $k,k' \in K$ and all other `mixed' commutators are zero\footnote{%
That is, commutators mixing $c$ and $e$ or mixing $J$ indices with $K$ indices are zero.}.
\end{definition}
%---------------------------------------------------
%---------------------------------------------------
\begin{definition}[Master space]
The master space corresponding to $I = J \sqcup K$ is
\[
   M = C(I) \oplus E
\]
with dimension $2n^2-2n-2$,
where $n = |I|$. Its Weyl algebra is therefore
\[
D_M = D_{C(I)} \otimes D_E
\]
We agree that elements of $D_{C(I)}$ and $D_E$
are extended to $D_M$. For example,
$\Xi$ is also viewed as an element of $D_M$ from now on.
Set $\Mp = \Cp \oplus E$,
so
\begin{equation}\label{eq:mdec}
M = \C \Xi \oplus \Mp
\end{equation}
\end{definition}
%--------------------------------------------
%--------------------------------------------
\begin{lemma} \label{lemma:pb}
There is a unique linear map $\down_J$ whose adjoint
$\down_J^\ast: C(\JB)^\ast \to M^\ast$
maps\footnote{%
On the right hand side, the variables $\kp$, $\cp$, $\ep$ are in $C(I)^\ast$
whereas the variables $c_{j\bullet}$, $e_{j\bullet}$ are in $E^\ast$.
}
\begin{align*}
k_{jj'} & \mapsto \kp_{jj'} &
c_{jj'} & \mapsto \cp_{jj'} &
e_{jj'} & \mapsto \ep_{jj'} \\
k_{j\bullet}, k_{\bullet j} & \mapsto - \textstyle\sum_{j' \in J} \kp_{j'j} &
c_{\bullet j} & \mapsto - \textstyle\sum_{j' \in J} \cp_{j'j}\\
& & c_{j\bullet} & \mapsto c_{j\bullet} &
e_{j\bullet}, e_{\bullet j} & \mapsto e_{j\bullet}
\end{align*}
for all $j,j' \in J$. Analogous for $\down_K$.
The maps $\down_J$, $\down_K$ are surjective.
\end{lemma}
\begin{proof}
We have specified $\down_J^\ast$ on a set that spans $C(\JB)^\ast$,
hence it exists and is unique if it respects all relations.
For instance, $\smash{\sum_{j' \in J} c_{j' j} + c_{\bullet j}}$ is zero in $C(\JB)^\ast$
and must be mapped to zero, which it is.
More interestingly,
$\smash{\sum_{j \in J} k_{\bullet j}}$ is zero in $C(\JB)^\ast$ and must be mapped to
zero, and it is by \eqref{eq:kperp}.
The last example shows that one must use $\kp$, not $k$, on $M^\ast$.
Surjectivity follows from the explicit right-inverses
constructed in Appendix \ref{app:kref}.
\qed\end{proof}
%--------------------------------------------
\begin{definition}[Contraction operator]\label{def:u}
Define $U \in D_M$ by
\[
U = 
\sum_{j \in J, k \in K}
\Big(
\kp_{jk} D_{c_{j\bullet}} D_{c_{k\bullet}}
+
\cp_{jk} D_{c_{j\bullet}} D_{e_{k\bullet}}
+
\cp_{kj} D_{e_{j\bullet}} D_{c_{k\bullet}}
+
\ep_{jk} D_{e_{j\bullet}} D_{e_{k\bullet}}
\Big)
\]
\end{definition}
%%%%%%%%%%%%%%%%%%%%%%%%%%%%%%%%%%%%%%%%%%%%%%%%%%%%%%%%%%%%%%%%%%%%%%%5
%%%%%%%%%%%%%%%%%%%%%%%%%%%%%%%%%%%%%%%%%%%%%%%%%%%%%%%%%%%%%%%%%%%%%%%%
%-------------------------------------------------
\section{Recursion for the amplitudes}\label{sec:rec}
Here we define tree scattering amplitudes for YM, denoted $A_I$,
and for GR, denoted $M_I$. 
The 3-point scattering amplitudes are polynomial.
\begin{definition}[Base case] \label{def:basecase}
If $|I| = 3$ define $A_I, M_I \in R_{C(I)}$ by
\begin{align*}
A_I & = c_{12}e_{31} + c_{23} e_{12} + c_{31} e_{23}\\
M_I & = (A_I)^2
\end{align*}
where $1,2,3$ are shorthands for the three elements of $I$.
\end{definition}
%-----------------------------
%
For higher amplitudes we state a recursion.
That this recursion admits a solution at all
is left as an assumption, Assumption \ref{asspa},
but a rough sketch of a proof is in Remark \ref{remark:existence}.
The purpose of this section is to state the recursion,
and to show that if a solution exists, then it is unique.
(See \cite{nr} for a self-contained treatment of amplitudes in $d=4$.)
\step
Roughly, the recursion says
that the amplitudes are rational functions
with only simple poles.
The poles are along a known family of hyperplanes,
and the residues are given recursively in terms of products of lower amplitudes.
This is known as factorization of residues.
The case $|I|=3$ serves as the base case.
\step
Define $X_i$, $Y_i$, $Z$ as in \eqref{eq:xyz}
using the differential operators in Lemma \ref{lemma:dops}.
\begin{definition}[Recursion] \label{def:rec}
By YM respectively GR tree amplitudes we mean a sequence
of rational functions, one for every integer $n \geq 3$.
Using an index set $I$ with $n = |I|$,
these rational functions are denoted, respectively,
\[
A_I, M_I \in \Frac(R_{C(I)})
\]
It is understood that $I$ must have a cyclic order for the YM amplitudes $A_I$
but no order for GR amplitudes $M_I$.
It is understood that $I$ is merely an index,
so if $I \simeq I'$ is a bijection,
preserving the cyclic order in the case of YM,
then $A_I \simeq A_{I'}$ respectively $M_I \simeq M_{I'}$
by relabeling coordinate indices. With these preliminaries, the following
properties must hold:
\begin{itemize}
\item \emph{Base case.} For $n=3$ the amplitudes are given by Definition \ref{def:basecase}.
\item \emph{Polynomiality in $c$ and $e$.}
The amplitudes $A_I$, $M_I$ are polynomial in the $c$- and $e$-variables.
So they are polynomials in the $c$- and $e$-variables with coefficients that are 
rational functions of the $k$-variables.
\item \emph{Permutation symmetry.} $A_I$ is invariant under cyclic
permutations of $I$, whereas $M_I$ is invariant under all permutations of $I$.
\item \emph{Gauge invariance.} For every $i \in I$, both $A_I$ and $M_I$ are annihilated by
$X_i$.
\item \emph{Polarization homogeneity.} For every $i \in I$
define a grading on $C(I)^\ast$ by
\[
|k_{ab}| = 0
\qquad
|c_{ab}| = \delta_{bi}
\qquad
|e_{ab}| = \delta_{ai} + \delta_{bi}
\]
for all $a \neq b$.
Then $A_I$ has degree $1$ and $M_I$ has degree $2$.
That is, they are annihilated by 
$Y_i$ with $h=1$ and $h=2$ respectively.
\item \emph{Momentum homogeneity.} Define a grading on $C(I)^\ast$ by
\[
|k_{ab}| = 2 \qquad
|c_{ab}| = 1 \qquad
|e_{ab}| = 0
\]
for all $a \neq b$.
Then $A_I$ has degree $4-n$ and
$M_I$ has degree $2$.
That is, they are annihilated by $Z$
with $s = 4-n$ and $s = 2$ respectively.
\item \emph{Poles and residues.}
Both $A_I$ and $M_I$ have poles only along the subspaces
$\xi = 0$ corresponding to decompositions $I = J \sqcup K$
with $|J|,|K| \geq 2$.
In the case of $A_I$, there are poles only for decompositions that respect the cyclic order,
thus a cyclic order is induced on $\JB$ and $\KB$ respectively.
The poles are simple\footnote{%
So the amplitudes are in the principal fractional ideal
generated by $1/(\prod \xi)$ where the product runs over all decompositions $I = J \sqcup K$
and $\xi$ is defined by \eqref{eq:defxi}.
} with residue
\begin{subequations}\label{eq:resid}
\begin{alignat}{6}\label{eq:rec}
\Res_{\xi = 0} A_I
\,&=&\;\;
\pm \text{const}\cdot U &
\down_{J}^\ast(A_\JB)
\down_{K}^\ast(A_\KB)
\Big|_{\xi = 0}\\
\Res_{\xi = 0} M_I
\,&=&\;\;
\pm \text{const}\cdot U^2 &
\down_{J}^\ast(M_\JB)
\down_{K}^\ast(M_\KB)
\Big|_{\xi = 0}
\end{alignat}
\end{subequations}
as an identity in $\Frac(R_M/\xi R_M)$.
The operator $U$ is in Definition \ref{def:u}.
The specification of the residue is recursive because $|\JB|, |\KB| < |I|$.
\end{itemize}
\end{definition}
In \eqref{eq:resid},
the role of $U$ respectively $U^2$ is to contract
the polarizations of the $\bullet$ particles
that the amplitudes on the right hand side depend on.
%----------------------------------------------------
\begin{remark}[Existence, rough sketch of proof] \label{remark:existence}
The YM amplitudes $A_I$ can be defined using color-ordered Feynman rules
in the Lorenz-Feynman gauge, see e.g.~\cite{dixon},
which yields rational functions of the products
$k_i \cdot k_j$, $k_i \cdot \eps_j$, $\eps_i \cdot \eps_j$ as required.
It is `well-known' that this has all the properties in Definition \ref{def:rec},
including gauge invariance and residue factorization.
The translation of residue factorization to the variables \eqref{eq:kce}
yields in particular the operator $U$ to carry out the
sum over polarizations of the $\bullet$ particles.
Given the $A_I$, the GR amplitudes $M_I$ can be defined using the field theoretical
Kawai-Lewellen-Tye or KLT relations,
see for instance \cite{sonder},
but then some properties in Definition \ref{def:rec} are not immediate,
for instance permutation symmetry and the factorization of residues.
These properties are discussed in some detail in the
`purely field theoretical view' of the KLT relations in \cite{sonder},
which exploits various known properties of the YM amplitudes and of the KLT kernels.
\end{remark}
%----------------------------------------------------
\begin{lemma}[Uniqueness] \label{lemma:uniq}
If amplitudes as in Definition \ref{def:rec} exist, then they are unique.
\end{lemma}
\begin{proof}
The proof is not sensitive to the
constants in \eqref{eq:resid},
part of which is a matter of normalization.
The proof is by induction on $n = |I|$.
The $n=3$ amplitudes are fixed.
For $n\geq 4$,
the residues of any two candidate amplitudes
are the same by \eqref{eq:resid} and by the induction hypothesis.
So the difference between any two candidates,
call it $u \in \Frac(R_{C(I)})$,
has no poles and is regular, except perhaps where the poles intersect.
But the union of all pairwise intersections of poles
has codimension two, and by the Hartogs extension theorem
(cf.~\cite{nr} for more comments about this in the context of $d=4$ amplitudes),
$u$ is actually globally regular.
This means that $u \in R_{C(I)}$, a polynomial.
Since $u$ is annihilated by $X_i$ and $Y_i$ and $Z$,
Lemma \ref{lemma:vanish} below implies $u = 0$.
\qed\end{proof}
%--------------------------------
Lemma \ref{lemma:vanish}
  below is not only used to prove uniqueness in Lemma \ref{lemma:uniq},
  it will also play a key role in our proof of Theorem \ref{theorem:main},
  hence the additional parameters $p_i$ and $q$.
\begin{lemma}[Vanishing lemma]\label{lemma:vanish}
Suppose $n = |I| \geq 4$
and suppose $X_i,Y_i,Z$ are as in \eqref{eq:xyz},
either with the parameters for YM or with the parameters for GR.
Suppose $p_i \geq 0$ with $i \in I$
are integers with $\sum_{i \in I} p_i < n$,
and $q \geq 0$ is another integer.
Suppose a polynomial $u \in R_{C(I)}$ satisfies:
\[
X_i u = 0
\qquad (Y_i + p_i) u = 0
\qquad (Z + q) u = 0
\]
for all $i \in I$. Then $u = 0$.
\end{lemma}
\begin{proof}
Distinguish the cases in Table \ref{table:vanish}.
The schematic structure refers to $u$ as a polynomial in the $c$, $k$ variables
with coefficients that are polynomials in the $e$ variables.
The schematic structure follows from $(Z+q)u = 0$.
The polynomials $P(e)$ and $Q(e)$ are homogeneous.
Their homogeneity degrees are given in terms
of $p = \sum_{i\in I} p_i < n$ and follow from $\sum_{i \in I} (Y_i + p_i)u = 0$.
Note that $P$ is a single homogeneous polynomial in case 3, whereas $P$ and $Q$ are
schematic for several homogeneous polynomials in cases 1 and 2.
\begin{table}
\begingroup
\renewcommand{\arraystretch}{1.4}
\begin{center}{\footnotesize
\begin{tabular}{r||r|l|l|l}
case & situation & schematic structure of $u$ & $\deg P$ & $\deg Q$\\
\hline\hline
1 & GR, $q = 0$ & $P(e)cc + Q(e)k$ & $n-\tfrac{p}{2}-1$ & $n-\tfrac{p}{2}$\\
2 & GR, $q = 1$ & $P(e)c$ & $n-\tfrac{p}{2}-\tfrac12$ \\
3 & GR, $q = 2$ & $P(e)$ &  $n-\tfrac{p}{2}$ \\
4 & GR, $q > 2$ & $0$ & \\
\hline
3 & YM, $n=4$ and $q=0$ & $P(e)$ & $\tfrac{n}{2} - \tfrac{p}{2} = 2-\tfrac{p}{2}$ &\\
4 & YM, $n>4$ or $q>0$ & 0  & 
\end{tabular}}
\end{center}
\caption{
List of cases in the proof of Lemma \ref{lemma:vanish}.
}\label{table:vanish}
\endgroup%
\end{table}
Since $p < n$ we have $\deg P > 0$, $\deg Q > 0$,
in particular $P$, $Q$ cannot be constant,
so they are zero if their derivatives are zero.
We discuss each case:
\begin{itemize}
\item Case 4: Here $u = 0$.
\item Case 3: Here $P(e)$ is a single polynomial.
Since $X_i u = 0$ for all $i$, we have
\[
\textstyle\sum_{j \in I} c_{ij} D_{e_{ij}} P = 0
\]
There is no sum over $i$.
For every fixed $i$ the only relations between the
$(c_{ij})_{j \in I}$ are $c_{ii} = 0$,
which suffices to conclude
$D_{e_{ij}} P = 0$. So $u=P=0$.
\item Case 2: Here $u = \sum_{a,b \in I} P_{ab} c_{ab}$
for some polynomials $P_{ab}  = P_{ab}(e)$.
Using \eqref{eq:linrel} we may assume $\sum_a P_{ab} = 0$ and $P_{aa} = 0$.
Since $X_i u = 0$ for all $i$,
\[
    \textstyle
    \sum_j k_{ji} P_{ji}
    +
    \sum_{a,b,j} c_{ab} c_{ij} D_{e_{ij}} P_{ab}
    = 0
\]
Since the terms scale differently in $c$ and $k$, they are separately zero.
Continue with $\sum_j k_{ji} P_{ji} = 0$. There is no sum over $i$.
Since $n \geq 4$, for every fixed $i$ the only relations between the
$(k_{ji})_{j \in I}$ are $k_{ii} = 0$ and $\sum_j k_{ji} = 0$, which implies $P_{ji} = 0$. So $u = 0$.
\item Case 1: Here $u = \sum_{a,b,c,d} P_{abcd} c_{ab}c_{cd} + \sum_{a,b} Q_{ab} k_{ab}$.
Using \eqref{eq:linrel}, we may assume $P_{abcd} = P_{cdab}$, $\sum_a P_{abcd} = 0$, 
$P_{aacd} = 0$, $Q_{ab} = Q_{ba}$, $\sum_a Q_{ab} = 0$, $Q_{aa} = 0$.
Using $X_iu=0$ we have
$V_i + W_i = 0$ using the abbreviations
\begin{align*}
V_i & = 
\textstyle
\sum_{a,b,j} (2k_{ji}c_{ab}P_{abji} + k_{ab} c_{ij} D_{e_{ij}} Q_{ab})\\
W_i & = 
\textstyle\sum_{a,b,c,d,j} c_{ab}c_{cd}c_{ij} D_{e_{ij}}P_{abcd}
\end{align*}
Since they scale differently in $c$ and $k$,
we have $V_i = W_i = 0$.
\begin{itemize}
\item It follows from $V_i = 0$
that (reasoning as in case 2) if all $P$ are zero then
all derivatives of all $Q$ are zero and then $Q=0$, since $\deg Q > 0$.
The problem is thus reduced to showing that all $P$ are zero.
\item
From $(1-\delta_{ai})D_{c_{ak}} V_i - (1-\delta_{ak})D_{c_{ai}} V_k = 0$ we get
(since this particular combination eliminates the $Q$ terms using $D_{e_{ik}} = D_{e_{ki}}$):
\begin{equation}\label{eq:n4}
\textstyle (1-\delta_{ai}) \sum_j k_{ji} P_{akji}
- (1-\delta_{ak}) \sum_j k_{jk} P_{aijk} = 0
\end{equation}
for all $a,i,k \in I$. By differentiating with respect to $D_{k_{cd}}$, one obtains
linear identities with constant coefficients for the $P$.
\item 
From $D_{c_{ab}}D_{c_{cd}}D_{c_{ef}}W_i = 0$ we get
\begin{multline}\label{eq:u1}
(1-\delta_{ab})
(\delta_{ai}-\tfrac{1}{n-1}) D_{e_{ib}} P_{cdef}
+
(1-\delta_{cd})
(\delta_{ci}-\tfrac{1}{n-1}) D_{e_{id}} P_{efab}\\
+
(1-\delta_{ef})
(\delta_{ei}-\tfrac{1}{n-1}) D_{e_{if}} P_{abcd}
=
0
\end{multline}
for all $a,b,c,d,e,f,i \in I$.
% Mathematica code to check this in random example
%n=6;
%DP[i_,i_,_,_,_,_]=0;
%DP[i_,j_,a_,b_,c_,d_]/;Not[OrderedQ[{i,j}]]:=DP[j,i,a,b,c,d];
%DP[i_,j_,a_,b_,c_,d_]/;Not[OrderedQ[{{a,b},{c,d}}]]:=DP[i,j,c,d,a,b];
%DP[_,_,a_,a_,_,_]=0;
%DP[_,_,_,_,a_,a_]=0;
%DP[i_,j_,n,b_,c_,d_]:=-Sum[DP[i,j,a,b,c,d],{a,1,n-1}];
%DP[i_,j_,n-1,n,c_,d_]:=-Sum[DP[i,j,a,n,c,d],{a,1,n-2}];
%DP[i_,j_,a_,b_,n,d_]:=-Sum[DP[i,j,a,b,c,d],{c,1,n-1}];
%DP[i_,j_,a_,b_,n-1,n]:=-Sum[DP[i,j,a,b,c,n],{c,1,n-2}];
%Dc[i_,j_][f_]:=(1-KroneckerDelta[i,j])*Sum[(1-KroneckerDelta[a,b])*(KroneckerDelta[i,a]-1/(n-1))*KroneckerDelta[j,b]*D[f,cc[a,b]],{a,1,n},{b,1,n}]//Expand;
%expr[i_]:=Sum[cc[a,b]*cc[c,d]*cc[i,j]*DP[i,j,a,b,c,d],{a,1,n},{b,1,n},{c,1,n},{d,1,n},{j,1,n}];
%pre[a_,b_,i_]:=(1-KroneckerDelta[a,b])*(KroneckerDelta[a,i]-1/(n-1));
%id[i_,a_,b_,c_,d_,e_,f_]:=pre[a,b,i]*DP[i,b,c,d,e,f]+pre[c,d,i]*DP[i,d,e,f,a,b]+pre[e,f,i]*DP[i,f,a,b,c,d];
%r:=RandomChoice[Range[1,n]];
%tt:=With[{i=r,a=r,b=r,c=r,d=r,e=r,f=r},
%id[i,a,b,c,d,e,f]-1/2*Composition[Dc[a,b],Dc[c,d],Dc[e,f]][expr[i]]]//Simplify
%Table[tt,{500}]
\end{itemize}
If $n \geq 5$ then \eqref{eq:u1} alone implies
that all derivatives of all $P$ are zero\footnote{%
Use \eqref{eq:u1} viewed as linear homogeneous identities
with constant coefficients for the $D_{e_{ij}} P_{abcd}$,
simplified using the algebraic conditions on $P$
and $D_{e_{ii}} = 0$, $D_{e_{ij}} = D_{e_{ji}}$.}
so $P=0$. If $n = 4$ then combining \eqref{eq:n4}, \eqref{eq:u1}
also yields $P = 0$. So $u=0$.
\end{itemize}
\qed\end{proof}
%------------------------------------
\section{Some commutators}\label{sec:comms}

Here we compute a number of commutators,
showing for instance that the $B_i$ appear in the commutators of type $[X_i,C_j]$.
These commutators are also used to prove Theorem \ref{theorem:main}.
In fact, Corollary \ref{corollary:xyz} contains identities for $A_j f$, $B_j f$, $C_j f$
with $f$ the amplitude,
and this will be used together with Lemma \ref{lemma:key} by which
$A_j f$, $B_j f$, $C_j f$ are pole-free,
and together with Lemma \ref{lemma:vanish},
to prove Theorem \ref{theorem:main}
which says that $A_j f$, $B_j f$, $C_j f$ vanish.

\begin{lemma}\label{lemma:comms}
Let $|I| \geq 3$. In $D_{C(I)}$ we have the following identities:
\begin{align*}
[X_i,A_j] & = -(1-\delta_{ij}) \tfrac{1}{|I|-1} D_{e_{ij}} X_j\\
[X_i,B_j] & = -2\delta_{ij} A_i
-(1-\delta_{ij})(
\textstyle\frac{2}{|I|-2} D_{c_{ji}} X_j
-D_{c_{ij}} X_i
-D_{e_{ij}} Y_i
+\frac{1}{|I|-1} D_{e_{ij}} (Y_j + Z))\\
[X_i,C_j] & =
- \delta_{ij} B_i + \tfrac{1}{|I|} B_i 
+ (1-\delta_{ij})(
  D_{k_{ij}} X_i
+ D_{c_{ji}} Y_i
- \tfrac{1}{|I|-2} D_{c_{ji}} Z)
\end{align*}
as well as
\begin{align*}
[Y_i,A_j] & = -2\delta_{ij} A_j &
[Z,A_j] & = 0 \\
[Y_i,B_j] & = -\delta_{ij} B_j &
[Z,B_j] & = -B_j\\
[Y_i,C_j] & = 0 &
[Z,C_j] & = -2C_j
\end{align*}
for all $i,j \in I$.
All these identities hold for both YM and GR;
recall that the parameters defining $Y_i$ and $Z$ are different in these two cases.
\end{lemma}
\begin{proof}
By direct calculation using
equations \eqref{eq:xyz}, \eqref{eqs:abc} and Lemma \ref{lemma:dops}.
\qed\end{proof}
%-------------------------------------------
\begin{corollary}\label{corollary:xyz}
Suppose $f$ is either the YM amplitude, $f = A_I$,
or the GR amplitude, $f = M_I$, as in Definition \ref{def:rec}. Then for all $i,j \in I$:
\begin{subequations}
\begin{align}
    \label{eq:acom}
    X_i(A_j f) & = 0 &
    (Y_i + 2\delta_{ij}) (A_j f) & = 0 &
     Z (A_j f) & = 0
\intertext{If $A_j f = 0$ for all $j$ then}
    \label{eq:bcom}
    X_i(B_jf) & = 0 &
    (Y_i + \delta_{ij}) (B_j f) & = 0 &
    (Z+1) (B_j f) & = 0
\intertext{If $B_j f = 0$ for all $j$ then}
   \label{eq:ccom}
   X_i (C_j f) & = 0 &
   Y_i(C_j f) & = 0 &
   (Z+2) (C_j f) & = 0
\end{align}
\end{subequations}
\end{corollary}
\begin{proof}
Use Lemma \ref{lemma:comms} and the fact that $X_i f = Y_i f = Z f = 0$ for all $i$.
\qed\end{proof}

%---------------------------------------------------------------------
%-----------------------------------------------------
\section{Proof of Theorem \ref{theorem:main}} \label{sec:proof1}
The annihilator of the YM amplitude $A_I$
is a left ideal in the Weyl algebra $D_{C(I)}$,
and likewise the annihilator of the GR amplitude $M_I$.
The elements $X_i$, $Y_i$, $Z$ are well-known to be in the annihilator\footnote{%
In the logic of this paper, this is enshrined in Definition \ref{def:rec}.},
and so are some constant coefficient operators of order 2 for YM, order 3 for GR,
that witness polynomiality in the $c$ and $e$ variables. Theorem
\ref{theorem:main} asserts that $A_i$, $B_i$, $C_i$ are also in the annihilator.
The proof, given at the end of this section, is by induction on $|I|$.
The induction step uses the following lemma.
%------------------------
\begin{lemma}[Key technical lemma] \label{lemma:key}
Let $U$ be as in Definition \ref{def:u}.
Suppose $I = J \sqcup K$ with $|J|, |K| \geq 2$ as before.
Suppose Theorem \ref{theorem:main}
holds for the index sets $\JB$ and $\KB$.
Cyclic orders are understood in the case of YM.
Then if
\[
	\mathbf{O} \in D_{C(I)} \subset D_M
\]
is one of $A_i$, $B_i$, $C_i$ with $i \in I$, then both
\begin{align*}
        \mathbf{O} U &\frac{\down^\ast_J(A_\JB) \down^\ast_K(A_\KB)}{\xi}\\
        \mathbf{O} U^2 &\frac{\down^\ast_J(M_\JB) \down^\ast_K(M_\KB)}{\xi}
\end{align*}
are elements of $\Frac(R_M)$ without pole along $\xi = 0$.
(More precisely,
 they are in the localization of $R_M$ at the codimension one prime ideal generated by $\xi$.)
\end{lemma}
%----------------------------------------------
\begin{proof}
The computation will be in the Weyl algebra,
we will not directly work with rational functions.
The first step is to split off the $\Xi$ direction,
using the direct sum decomposition \eqref{eq:mdec}. 
At the Weyl algebra level,
\[
D_M = D_{\C \Xi} \otimes D_{\Mp}
\]
Here $D_{\C \Xi}$ is the Weyl algebra generated by $\xi$ and $\Xi$,
with $[\Xi,\xi] = 2$.
Every element of $D_{\C \Xi}$ commutes with every element of $D_{\Mp}$.
For every $\mathbf{O}$ we have
\begin{align*}
\mathbf{O} & = 
\begin{pmatrix}
1 & \xi
\end{pmatrix}
\begin{pmatrix}
S_{00} & S_{01} & S_{02}\\
S_{10} & S_{11} & S_{12}
\end{pmatrix}
\begin{pmatrix}
1\\
\Xi\\
\Xi^2
\end{pmatrix}
\end{align*}
for unique $S_{00}, \ldots, S_{12} \in D_{\Mp}$. 
This isolates all occurrences of $\xi$ and $\Xi$.
It is clear from \eqref{eqs:abc} that at most one $\xi$ and at most two $\Xi$ appear.
Set
\begin{align*}
S_0 & = S_{10}\\
S_1 & = S_{00} - 2S_{11} \displaybreak[0]\\
S_2 & = S_{01} - 4S_{12}\\
S_3 & = S_{02}
\end{align*}
Note that $\Xi$ commutes with $U \in D_{\Mp}$ from Definition \ref{def:u},
and the pushforward of $\Xi$ under both $\down_J$ and $\down_K$ is zero.
Hence it suffices to show that
\begin{subequations}
\begin{align}
\label{eq:ym1}
S_1 U & \down^\ast_J(A_\JB) \down^\ast_K(A_\KB) = 0\\
\label{eq:ym2}
S_2 U & \down^\ast_J(A_\JB) \down^\ast_K(A_\KB) = 0\\
S_3 U & \down^\ast_J(A_\JB) \down^\ast_K(A_\KB) = 0
\end{align}
\end{subequations}
for YM, corresponding (by definition of $S_0,\ldots,S_3$)
to the absence of $1/\xi$ and $1/\xi^2$
and $1/\xi^3$ terms respectively,
and analogously it suffices to show that
\begin{subequations}
\begin{align}
\label{eq:gr1}
S_1 U^2 & \down^\ast_J(M_\JB) \down^\ast_K(M_\KB) = 0\\
\label{eq:gr2}
S_2 U^2 & \down^\ast_J(M_\JB) \down^\ast_K(M_\KB) = 0\\
S_3 U^2 & \down^\ast_J(M_\JB) \down^\ast_K(M_\KB) = 0
\end{align}
\end{subequations}
for GR. Note that $S_0$ has dropped out of the computation since
it cannot generate a pole along $\xi = 0$. Note that:
\begin{itemize}
\item $S_2 = S_3 = 0$ for $\mathbf{O} = A_i$
because it involves no $k$-derivative,
hence no $\Xi$-derivatives.
This also implies $S_{11} = 0$ in this case.
\item $S_3 = 0$ for $\mathbf{O} = B_i$
because it involves no second $k$-derivatives,
hence no second $\Xi$ derivatives.
This also implies $S_{12} = 0$ in this case.
\item $S_3 = 0$ for $\mathbf{O} = C_i$ but this requires a calculation.
It suffices to show that the analogous claim holds
for $\widetilde{C}_i$, that is for the following terms in \eqref{eqs:c}:
\[
\textstyle
a\sum_{j,k \in I} k_{jk} D_{k_{ij}} D_{k_{ik}}
+
b\sum_{j,k \in I} k_{jk} D_{k_{ji}} D_{k_{jk}}
\]
with $a = \tfrac12$ and $b = -1$.
Use Definition \ref{def:xixi} to replace
$k_{ab} = \kp_{ab} + \text{const}_{ab} \cdot \xi$
and $D_{k_{ab}} = \Dp_{k_{ab}} + \text{const}_{ab} \cdot \Xi$
and keep only the $\kp_{ab}$ respectively $\Xi$ terms. 
Using the first equation in \eqref{eq:kperp}
and $2a+b = 0$, we get zero.
\end{itemize}
It now suffices to show
\begin{subequations}\label{eq:suff1}
\begin{align}
&\text{\eqref{eq:ym1}, \eqref{eq:gr1}
for $A_j, B_j, C_j$ for all $j \in J$.}\\
&\text{\eqref{eq:ym2},
 \eqref{eq:gr2}
for $B_j, C_j$ for all $j \in J$.}
\end{align}
\end{subequations}
The restriction to $i = j \in J$ is new and is without loss of generality.
It entails that the rest of this proof is not symmetric
under exchanging $J$ and $K$.
The rest of this proof exploits the direct sum decomposition
of $\Mp$ in \eqref{eq:decomp},
see also Corollary \ref{corollary:decomp}. At Weyl algebra level,
\[
D_{\Mp} = D_{\image \pi}
          \otimes \underbrace{D_{\image \up_J}}_{\simeq D_{C(\JB)}}
          \otimes \underbrace{D_{\image \up_K}}_{\simeq D_{C(\KB)}}
\]
The indicated isomorphisms
are established by $\down_J$, $\up_J$ and $\down_K$, $\up_K$.
Call them
\begin{align*}
\pbc_J : D_{C(\JB)} &\to D_{\image \up_J}\\
\pbc_K : D_{C(\KB)} &\to D_{\image \up_K}
\end{align*}
They are explicitly given by the formulas in Lemmas \ref{lemma:pb}, \ref{lemma:pushf}, \ref{lemma:rinvp}.
Let
$I_{\text{YM}} \subset D_{\Mp}$
respectively
$I_{\text{GR}} \subset D_{\Mp}$
be the left ideals generated by
(the difference between YM and GR is implicit in the parameters defining the $Y$ and $Z$ elements):
\begin{itemize}
\item The left ideal $I_\pi \subset D_{\image \pi}$ generated by all partial derivatives.
Equivalently, this is the annihilator of the constant functions on $\image \pi$.
\item The left ideal in $D_{\image \up_J}$ 
generated by the image under $\pbc_J$ of the following elements,
which are known annihilators
of $A_\JB$ respectively $M_\JB$:
\begin{subequations}\label{eq:aan}
\begin{equation}
A_j, A_\bullet, B_j, B_\bullet, C_j, C_\bullet, X_j, X_\bullet, Y_j, Y_\bullet, Z
\;\;\in\;\;
D_{C(\JB)}
\end{equation}
and\footnote{These are also known annihilators. They witness polynomiality in the $c$ and $e$ variables
related to `the polarization of particle $\bullet$'
with degrees fixed by $Y_{\bullet} \in D_{C(\JB)}$.}
\begin{align}\label{eq:aanpoly}
&\text{for YM:}\;\;\; D_{x_{j_1\bullet}} D_{y_{j_2\bullet}}
&&\text{for GR:}\;\;\; D_{x_{j_1\bullet}} D_{y_{j_2\bullet}} D_{z_{j_3\bullet}}
\end{align}
\end{subequations}
for all $x,y,z \in \{c,e\}$ and all $j_1,j_2,j_3 \in J$.
\item The left ideal in $D_{\image \up_K}$
generated by the image under $\pbc_K$ of the following elements,
which are known annihilators of $A_\KB$ respectively $M_\KB$:
\begin{subequations}\label{eq:ban}
\begin{equation}
A_\bullet, B_\bullet, C_\bullet, X_\bullet, Y_\bullet, Z
\;\;\in\;\;
D_{C(\KB)}
\end{equation}
and
\begin{align}\label{eq:banpoly}
&\text{for YM:}\;\;\; D_{x_{k_1\bullet}} D_{y_{k_2\bullet}}
&&\text{for GR:}\;\;\; D_{x_{k_1\bullet}} D_{y_{k_2\bullet}} D_{z_{k_3\bullet}}
\end{align}
\end{subequations}
for all $x,y,z \in \{c,e\}$ and all $k_1,k_2,k_3 \in K$.
\end{itemize}
It is part of the assumptions of this lemma,
and a consequence of Definition \ref{def:rec}, that \eqref{eq:aan} and \eqref{eq:ban}
are known annihilators of $A_\JB$, $M_\JB$ and $A_\KB$, $M_\KB$.
More annihilators are known, but we will not need those.
To show \eqref{eq:suff1} it now suffices to show that
(by construction of $I_{\text{YM}}$, $I_{\text{GR}}$
and using $\down_J \pi = 0$, $\down_K \pi = 0$):
\begin{subequations}\label{eq:susu}
\begin{align}
\label{eq:susuym1}
S_1 U & \in I_{\text{YM}} && \text{for $\mathbf{O} = A_j, B_j, C_j$}\\
\label{eq:susuym2}
S_2 U & \in I_{\text{YM}} && \text{for $\mathbf{O} = B_j, C_j$} \displaybreak[0]\\
\label{eq:susugr1}
S_1 U^2 & \in I_{\text{GR}} && \text{for $\mathbf{O} = A_j, B_j, C_j$}\\
\label{eq:susugr2}
S_2 U^2 & \in I_{\text{GR}} && \text{for $\mathbf{O} = B_j, C_j$}
\end{align}
\end{subequations}
Thus the problem is reduced to one of
checking membership in a left ideal in the Weyl algebra $D_{\Mp}$.
(For every fixed $|I|$ this can in principle be checked
 algorithmically using Gr\"obner bases.
 But we need to prove membership for all $|I|$.)
To proceed, we use the canonical $D_{\image \pi}/I_\pi \simeq R_{\image \pi}$.
Here $R_{\image \pi}$ are the polynomials on $\image \pi$.
This gives a canonical map
\[
\rho : D_{\Mp} \to {\underbrace{R_{\image \pi}}_{\simeq D_{\image \pi} / I_\pi}}
           \otimes D_{\image \up_J}
           \otimes D_{\image \up_K}
\]
Decompose $R_{\image \pi} = \textstyle\bigoplus_{d \geq 0} (R_{\image \pi})_d$
where $(R_{\image \pi})_d$ are all polynomials that are homogeneous of degree $d$.
Accordingly $\rho = \bigoplus_{d \geq 0} \rho_d$ where
\[
\rho_d : D_{\Mp} \to (R_{\image \pi})_d
           \otimes D_{\image \up_J}
           \otimes D_{\image \up_K}
\]
Clearly $\rho_d(S_1)$, $\rho_d(S_2)$ can only be nonzero for $d = 0,1$
because all $\mathbf{O}$ in \eqref{eqs:abc} have polynomial coefficients of order $\leq 1$.
We claim that actually
\begin{equation}\label{eq:rsrs}
\rho_1(S_1) = \rho_1(S_2) = 0
\end{equation}
for all $\mathbf{O} = A_j, B_j, C_j$ (an analogous statement fails for $\widetilde{C}_j$).
To see this, use the description of $\image \pi$ in Corollary \ref{corollary:decomp}.
Here are some examples:
\begin{itemize}
\item Consider the term
$\smash{\sum_{a,b \in I} e_{ab} D_{e_{aj}} D_{e_{bj}}}$ in $A_j$.
This term is already in $D_\Mp$, since $e = \smash{\ep}$ and
$D_e = \smash{\Dp_e}$.
If $a,b \in J$ or $a,b \in K$ then we get no contribution
to $\rho_1$ since $\ep_{ab}$ commutes with all elements of $\image \pi$
by Corollary \ref{corollary:decomp}.
If $a \in K$ then $\smash{\Dp_{e_{aj}}} \in I_\pi$
and if $b \in K$ then $\smash{\Dp_{e_{bj}}} \in I_\pi$
and we also get no contribution to $\rho_1$.
\item Consider next $\sum_{a,b \in I} c_{ab} D_{c_{aj}} D_{e_{bj}}$ in $A_j$.
By the same reasoning, it suffices to consider the sum over $a \in K$, $b \in J$.
For $a \in K$, write $D_{c_{aj}} = \Dp_{c_{aj}}$ as
\[
\Dp_{c_{aj}} = \textstyle(\Dp_{c_{aj}} - r) + r
\qquad \text{where} \qquad r = \frac{1}{|K|} \sum_{k' \in K} \Dp_{c_{k'j}}
\]
The first term does not contribute to $\rho_1$ since $\Dp_{c_{aj}} - r \in I_\pi$
by Corollary \ref{corollary:decomp}.
The second does not contribute to $\rho_1$ since
$\smash{\sum_{a \in K} \cp_{ab} = - \sum_{a \in J} \cp_{ab}}$
which for $b \in J$ commutes with all elements of $\image \pi$
by Corollary \ref{corollary:decomp}. 
\item Terms involving $k$ or $D_k$ are more complicated.
One must take into account $k = \kp + \text{const}\cdot \xi$
and $D_k = \Dp_k + \text{const} \cdot \Xi$, in Definition \ref{def:xixi}.
\end{itemize}
%--------------------------------------------------
Using \eqref{eq:rsrs}
one can see that $\rho_d(S_1 U)$, $\rho_d(S_2 U)$ can only be nonzero for $d = 0,1$
and that $\rho_d(S_1 U^2)$, $\rho_d(S_2 U^2)$ can only be nonzero for $d = 0,1,2$.
In each of these cases, and for $\mathbf{O} = A_j, B_j, C_j$,
Tables \ref{table:ym} and \ref{table:gr} list elements
in \eqref{eq:aan} and \eqref{eq:ban}
that suffice to prove membership in $I_{\text{YM}}$ respectively $I_{\text{GR}}$,
as in \eqref{eq:susu}.
The tag `$\polyann$' subsumes all annihilators in \eqref{eq:aanpoly}, \eqref{eq:banpoly}
that witness polynomiality.

%------------------------------------------------
\newcommand{\sbx}{\rule{30mm}{0pt}}
\begin{table}
\centering
{\footnotesize
\begin{tabular}{c|c||l|l}
\sbx & $d$ & elements from \eqref{eq:aan} & elements from \eqref{eq:ban} \\
\hline
$\rho_d(S_1U)$ for $\mathbf{O} = A_j$
&   $0$ & & $X_{\bullet}$ \\
&   $1$ & $A_j$, $X_j$, $\polyann$ & \\
\hline
$\rho_d(S_2U)$ for $\mathbf{O} = B_j$
&   $0$ & & $X_{\bullet}$\\
&   $1$ & $X_j$, $\polyann$\\
\hline
$\rho_d(S_1U)$ for $\mathbf{O} = B_j$
&   $0$ & & $B_{\bullet}$, $X_{\bullet}$, $\polyann$\\
&   $1$ & $B_j$, $X_j$, $Y_j$, $Z$, $\polyann$ & $Z$, $\polyann$\\
\hline
$\rho_d(S_2U)$ for $\mathbf{O} = C_j$
&   $0$ & & $X_{\bullet}$\\
&   $1$ & $Z$, $\polyann$ & $Z$, $\polyann$\\
\hline
$\rho_d(S_1U)$ for $\mathbf{O} = C_j$
&   $0$ & & $B_{\bullet}$, $X_{\bullet}$, $\polyann$\\
&   $1$ & $C_j$, $C_{\bullet}$, $Z$, $\polyann$ & $C_{\bullet}$, $Z$, $\polyann$
\end{tabular}}
\caption{For YM,
this table lists elements
in \eqref{eq:aan} and \eqref{eq:ban}
that suffice to prove membership in $I_{\text{YM}}$
as in \eqref{eq:susuym1}, \eqref{eq:susuym2}. The list is not necessarily minimal.}\label{table:ym}
\end{table}
%-----------------------------------------------
\begin{table}
\centering
{\footnotesize
\begin{tabular}{c|c||l|l}
\sbx & $d$ & elements from \eqref{eq:aan} & elements from \eqref{eq:ban} \\
\hline
$\rho_d(S_1U^2)$ for $\mathbf{O} = A_j$
&   $0$ & & $A_{\bullet}$ \\
&   $1$ & & $X_{\bullet}$ \\
&   $2$ & $A_j$, $X_j$, $\polyann$ & \\
\hline
$\rho_d(S_2U^2)$ for $\mathbf{O} = B_j$
&   $0$ & & \\
&   $1$ & & $X_{\bullet}$\\
&   $2$ & $X_j$, $\polyann$\\
\hline
$\rho_d(S_1U^2)$ for $\mathbf{O} = B_j$
&   $0$ & $X_{\bullet}$ & $A_{\bullet}$\\
&   $1$ & & $B_{\bullet}$, $X_{\bullet}$, $\polyann$\\
&   $2$ & $B_j$, $X_j$, $Y_j$, $Z$, $\polyann$ & $Z$, $\polyann$\\
\hline
$\rho_d(S_2U^2)$ for $\mathbf{O} = C_j$
&   $0$ & & \\
&   $1$ & & $X_{\bullet}$\\
&   $2$ & $Z$, $\polyann$ & $Z$, $\polyann$\\
\hline
$\rho_d(S_1U^2)$ for $\mathbf{O} = C_j$
&   $0$ & $A_{\bullet}$, $X_{\bullet}$ & $A_{\bullet}$\\
&   $1$ & & $B_{\bullet}$, $X_{\bullet}$, $\polyann$\\
&   $2$ & $C_j$, $C_{\bullet}$, $Z$, $\polyann$ & $C_{\bullet}$, $Z$, $\polyann$
\end{tabular}}
\caption{For GR,
this table lists elements
in \eqref{eq:aan} and \eqref{eq:ban}
that suffice to prove membership in $I_{\text{GR}}$
as in \eqref{eq:susugr1}, \eqref{eq:susugr2}. The list is not necessarily minimal.}\label{table:gr}
\end{table}
We now discuss these tables in detail.
The identities below are in 
\begin{equation}\label{eq:hos}
(R_{\image \pi})_0
           \otimes D_{\image \up_J}
           \otimes D_{\image \up_K}
\;\simeq\;
                   D_{\image \up_J}
           \otimes D_{\image \up_K}
\end{equation}
To extract the various `Taylor coefficients'
we use the $\Dw$ defined in Corollary \ref{corollary:decomp},
understood here as mapping $(R_{\image \pi})_d \to (R_{\image \pi})_{d-1}$.
We will not make explicit the `$\polyann$' pieces
and state some identities in the schematic form
\[
a = b \modp
\]
which asserts that $a-b$
is in the left ideal of \eqref{eq:hos} generated, via $\pbc_J$ and $\pbc_K$ respectively,
by \eqref{eq:aanpoly} and \eqref{eq:banpoly}.
With these preliminaries,
one row in Table \ref{table:gr} with $\mathbf{O} = C_j$ is
proved by the identities
\begin{equation}\label{eq:id4}
\begin{aligned}
\Dw_{k_{j_1k_1}} \rho_1(S_2 U^2) \;&=\;
2D_{c_{j_1\bullet}} D_{c_{k_1\bullet}} D_{c_{j \bullet}} \pbc_K(X_{\bullet})\\
\Dw_{c_{j_1k_1}} \rho_1(S_2 U^2) \;&=\;
2D_{c_{j_1\bullet}} D_{e_{k_1\bullet}} D_{c_{j \bullet}} \pbc_K(X_{\bullet})\\
\Dw_{c_{k_1j_1}} \rho_1(S_2 U^2) \;&=\;
2D_{e_{j_1\bullet}} D_{c_{k_1\bullet}} D_{c_{j \bullet}} \pbc_K(X_{\bullet})\\
\Dw_{e_{j_1k_1}} \rho_1(S_2 U^2) \;&=\;
2D_{e_{j_1\bullet}} D_{e_{k_1\bullet}} D_{c_{j \bullet}} \pbc_K(X_{\bullet})
\end{aligned}
\end{equation}
in the space \eqref{eq:hos},
they hold for all $j_1 \in J$ and $k_1 \in K$.
It is essential here that derivatives such as $D_{e_{k_1\bullet}}$ are to the left
of $\pbc_K(X_{\bullet})$.
On the other hand,
since  $\pbc_K(D_{e_{k_1\bullet}}) = D_{e_{k_1\bullet}}$
by Lemmas \ref{lemma:pushf} and \ref{lemma:rinvp},
it does not matter if this derivative is written inside or outside of $\pbc_K$.
We abbreviate
\begingroup
\renewcommand{\arraystretch}{1.15}
\begin{align*}
\Dw_1 & = \begin{pmatrix}
\Dw_{k_{j_1k_1}}\\
\Dw_{c_{j_1k_1}}\\
\Dw_{c_{k_1j_1}}\\
\Dw_{e_{j_1k_1}}
\end{pmatrix}
&
R_1 & = \begin{pmatrix}
D_{c_{j_1\bullet}} D_{c_{k_1\bullet}}\\
D_{c_{j_1\bullet}} D_{e_{k_1\bullet}}\\
D_{e_{j_1\bullet}} D_{c_{k_1\bullet}}\\
D_{e_{j_1\bullet}} D_{e_{k_1\bullet}}
\end{pmatrix}
\intertext{and}
\Dw_2 & = \begin{pmatrix}
\Dw_{k_{j_1k_1}}\Dw_{k_{j_2k_2}}\\
\Dw_{c_{j_1k_1}}\Dw_{c_{j_2k_2}}\\
\Dw_{c_{k_1j_1}}\Dw_{c_{k_2j_2}}\\
\Dw_{e_{j_1k_1}}\Dw_{e_{j_2k_2}}\\
\Dw_{k_{j_1k_1}}\Dw_{c_{j_2k_2}}\\
\Dw_{c_{j_1k_1}}\Dw_{c_{k_2j_2}}\\
\Dw_{c_{k_1j_1}}\Dw_{e_{j_2k_2}}\\
\Dw_{k_{j_1k_1}}\Dw_{c_{k_2j_2}}\\
\Dw_{c_{j_1k_1}}\Dw_{e_{j_2k_2}}\\
\Dw_{k_{j_1k_1}}\Dw_{e_{j_2k_2}}
\end{pmatrix}
& R_2 & =
\begin{pmatrix}
D_{c_{j_1\bullet}} D_{c_{j_2\bullet}} D_{c_{k_1\bullet}} D_{c_{k_2\bullet}}\\
D_{c_{j_1\bullet}} D_{c_{j_2\bullet}} D_{e_{k_1\bullet}} D_{e_{k_2\bullet}}\\
D_{e_{j_1\bullet}} D_{e_{j_2\bullet}} D_{c_{k_1\bullet}} D_{c_{k_2\bullet}}\\
D_{e_{j_1\bullet}} D_{e_{j_2\bullet}} D_{e_{k_1\bullet}} D_{e_{k_2\bullet}}\\
D_{c_{j_1\bullet}} D_{c_{j_2\bullet}} D_{c_{k_1\bullet}} D_{e_{k_2\bullet}}\\
D_{c_{j_1\bullet}} D_{e_{j_2\bullet}} D_{e_{k_1\bullet}} D_{c_{k_2\bullet}}\\
D_{e_{j_1\bullet}} D_{e_{j_2\bullet}} D_{c_{k_1\bullet}} D_{e_{k_2\bullet}}\\
D_{c_{j_1\bullet}} D_{e_{j_2\bullet}} D_{c_{k_1\bullet}} D_{c_{k_2\bullet}}\\
D_{c_{j_1\bullet}} D_{e_{j_2\bullet}} D_{e_{k_1\bullet}} D_{e_{k_2\bullet}}\\
D_{c_{j_1\bullet}} D_{e_{j_2\bullet}} D_{c_{k_1\bullet}} D_{e_{k_2\bullet}}
\end{pmatrix}
\end{align*}
\endgroup
with the understanding that $j_1,j_2 \in J$ and $k_1,k_2 \in K$.
The four identities \eqref{eq:id4} are now given, more succinctly, by
\[
\Dw_1 \rho_1(S_2 U^2) =
2 D_{c_{j \bullet}} R_1 \pbc_K(X_{\bullet})
\]
We now state all identities needed
for \eqref{eq:susu}.
For YM, with reference to Table \ref{table:ym}:
\begin{itemize}
\item If $\mathbf{O} = A_j$:
\begin{align*}
\rho_0(S_1 U) & = 
   - \tfrac{|I|}{|I|-1} \pbc_J(D_{c_{\bullet j}}) D_{e_{j\bullet}} \pbc_K(X_\bullet)  \displaybreak[0] \\
%-----------
\Dw_1 \rho_1(S_1 U) & = 
   R_1 \pbc_J(A_j)
   + \tfrac{|K|-1}{|I|-1} \pbc_J(D_{c_{\bullet j}}) R_1 \pbc_J(X_j)
   \modp 
\end{align*}
%--------------
\item If $\mathbf{O} = B_j$:
\begin{align*}
\rho_0(S_2 U) & = 
  \tfrac{(|J|-1)(|I|^2-|I|-|J|)}{|J|(|I|-1)(|I|-2)} D_{e_{j\bullet}} \pbc_K(X_\bullet)  \displaybreak[0] \\
%-----
\Dw_1 \rho_1(S_2 U) & = 
   -\tfrac{|K|(|K|-1)}{(|I|-1)(|I|-2)} R_1 \pbc_J(X_j)
   \modp
%-------
\intertext{and}
\rho_0(S_1 U) & = 
   D_{e_{j\bullet}} \pbc_K(B_\bullet)
   - \tfrac{1}{|J|(|I|-2)}
   \pbc_J\big(\\
   & \qquad (|I|-2)(|J|+1) D_{c_{j\bullet}} D_{c_{\bullet j}} \\
   & \qquad + |I| (|J|-1) D_{e_{j\bullet}} D_{k_{\bullet j}} \\
   & \qquad + (|I|-2) \textstyle\sum_{j' \in J} D_{c_{j'\bullet}} D_{c_{j'j}}
    \big) \pbc_K(X_\bullet)
    \modp  \displaybreak[0] \\
%--------------
\Dw_1 \rho_1(S_1 U) & = 
    R_1 \pbc_J(B_j)
     + \tfrac{2(|K|-1)}{|I|-2} \pbc_J(D_{k_{\bullet j}}) R_1 \pbc_J(X_j)\\
     & \qquad
     + \tfrac{|K|-1}{|I|-1} \pbc_J(D_{c_{\bullet j}}) R_1 \pbc_J(Y_j)\\
     & \qquad
     + \tfrac{|K|-1}{|I|-1} \pbc_J(D_{c_{\bullet j}}) R_1 \pbc_J(Z)\\
     & \qquad
     - \tfrac{|J|}{|I|-1} \pbc_J(D_{c_{\bullet j}}) R_1 \pbc_K(Z)
    \modp
\end{align*}
%------------------------
\item If $\mathbf{O} = C_j$:
\begin{align*}
\rho_0(S_2 U) & = 
   D_{c_{j\bullet}} \pbc_K(X_\bullet)  \displaybreak[0] \\
\Dw_1 \rho_1(S_2 U) & = 
   -\tfrac{|K|(|K|-1)}{|I|(|I|-2)} R_1 \pbc_J(Z)
   +\tfrac{|K|(|J|-1)}{|I|(|I|-2)} R_1 \pbc_K(Z)
   \modp
%--------------------
\intertext{and}
\rho_0(S_1 U) & = 
  D_{c_{j\bullet}} \pbc_K(B_\bullet)
  -\tfrac{1}{|J|} \pbc_J\big(\\
    & \qquad
       (|J|+1) D_{c_{j\bullet}} D_{k_{\bullet j}}\\
    & \qquad + \textstyle\sum_{j' \in J} D_{c_{j'\bullet}} D_{k_{jj'}}\\
    & \qquad - \textstyle\sum_{j' \in J} D_{c_{j'\bullet}} D_{k_{\bullet j'}}
    \big) \pbc_K(X_\bullet)
  \modp   \displaybreak[0] \\
%-------------------
\Dw_1 \rho_1(S_1 U) & = 
  -\tfrac{|K|-1}{|I|} R_1 \pbc_J(C_\bullet)
  +R_1 \pbc_J(C_j)\\
  & \qquad
  +\tfrac{|K|-1}{|I|-2} \pbc_J(D_{k_{\bullet j}}) R_1 \pbc_J(Z)
  +\tfrac{|K|+1}{|I|} R_1 \pbc_K(C_\bullet)\\
  & \qquad
  -\tfrac{|J|-1}{|I|-2} \pbc_J(D_{k_{\bullet j}}) R_1 \pbc_K(Z)
  \modp
\end{align*}
\end{itemize}
%%%%%%%%%%%%%%%%%%%%%%%%%%%%%%%%%%%%%%%%%%%%%%%%%%%%%%%
For GR, with reference to Table \ref{table:gr}:
\begin{itemize}
\item If $\mathbf{O} = A_j$:
\begin{align*}
\rho_0(S_1 U^2) & =
   2 D_{e_{j\bullet}}^2 \pbc_K(A_{\bullet})  \displaybreak[0] \\
%----
\Dw_1 \rho_1(S_1U^2) & = 
   -2 \tfrac{|I|}{|I|-1}  \pbc_J(D_{c_{\bullet j}}) D_{e_{j\bullet}} R_1 \pbc_K(X_{\bullet})  \displaybreak[0] \\
%----
\Dw_2 \rho_2(S_1U^2) & = 
   2 R_2 \pbc_J(A_j) + 2 \tfrac{|K|-1}{|I|-1} \pbc_J(D_{c_{\bullet j}}) R_2 \pbc_J(X_j)
   \modp
\end{align*}
%%%%%%%%%%%%%%%%%%%%%%%%%%%%%%%%%%
\item If $\mathbf{O} = B_j$:
\begin{align*}
\rho_0(S_2 U^2) & = 0  \displaybreak[0] \\
%-------------
\Dw_1 \rho_1(S_2 U^2) & =
   2 \tfrac{(|J|-1)(|I|^2-|I|-|J|)}{|J|(|I|-1)(|I|-2)} D_{e_{j\bullet}} R_1 \pbc_K(X_\bullet)  \displaybreak[0] \\
%-----------
\Dw_2 \rho_2(S_2 U^2) & = 
   -2 \tfrac{|K|(|K|-1)}{(|I|-1)(|I|-2)} R_2 \pbc_J(X_j)
   \modp
\intertext{and}
\rho_0(S_1 U^2) & =
   -\tfrac{2}{|K|} D_{e_{j\bullet}} \textstyle\sum_{k'\in K} (D_{c_{k'\bullet}})^2 \pbc_J(X_\bullet)
   + 4 D_{c_{j\bullet}} D_{e_{j\bullet}} \pbc_K(A_\bullet)  \displaybreak[0] \\
%--------
\Dw_1 \rho_1(S_1 U^2) & =
   2 D_{e_{j\bullet}} R_1 \pbc_K(B_\bullet)
   - \tfrac{2}{|J|(|I|-2)}
   \pbc_J\big(\\
   & \qquad (|I|-2)(|J|+1) D_{c_{j\bullet}} D_{c_{\bullet j}} \\
   & \qquad + |I| (|J|-1) D_{e_{j\bullet}} D_{k_{\bullet j}} \\
   & \qquad + (|I|-2) \textstyle\sum_{j' \in J} D_{c_{j'\bullet}} D_{c_{j'j}}
    \big) R_1 \pbc_K(X_\bullet)
   \modp  \displaybreak[0] \\
%----------
\Dw_2 \rho_2(S_1U^2) & = 
    2 R_2 \pbc_J(B_j) 
    + \tfrac{4(|K|-1)}{|I|-2} \pbc_J(D_{k_{\bullet j}})  R_2 \pbc_J(X_j)\\
  & \qquad + \tfrac{2(|K|-1)}{|I|-1} \pbc_J(D_{c_{\bullet j}}) R_2 \pbc_J(Y_j)\\
  & \qquad + \tfrac{2(|K|-1)}{|I|-1} \pbc_J(D_{c_{\bullet j}}) R_2 \pbc_J(Z)\\
  & \qquad - \tfrac{2|J|}{|I|-1} \pbc_J(D_{c_{\bullet j}}) R_2 \pbc_K(Z)
    \modp
\end{align*}
%--------------------------
\item If $\mathbf{O} = C_j$:
\begin{align*}
\rho_0(S_2 U^2) & = 0  \displaybreak[0] \\
%-------
\Dw_1 \rho_1(S_2 U^2) & =
   2 D_{c_{j\bullet}} R_1 \pbc_K(X_\bullet)  \displaybreak[0] \\
%------
\Dw_2 \rho_2(S_2 U^2) & = 
   -\tfrac{2|K|(|K|-1)}{|I|(|I|-2)} R_2 \pbc_J(Z)
   +\tfrac{2|K|(|J|-1)}{|I|(|I|-2)} R_2 \pbc_K(Z)
   \modp
\intertext{and}
\rho_0(S_1 U^2) & = 
   - \tfrac{2}{|I|} \textstyle\sum_{k'\in K} D_{c_{k'\bullet}}^2 \pbc_J(A_\bullet)
   - \tfrac{2}{|K|} D_{c_{j\bullet}} \textstyle\sum_{k'\in K} D_{c_{k'\bullet}}^2 \pbc_J(X_\bullet)\\
   & \qquad
    + 2(D_{c_{j\bullet}}^2 - \tfrac{1}{|I|} \textstyle\sum_{j' \in J} D_{c_{j'\bullet}}^2) \pbc_K(A_\bullet)  \displaybreak[0] \\
%---------
\Dw_1 \rho_1(S_1 U^2) & = 
   2 D_{j\bullet} R_1 \pbc_K(B_\bullet)
   - \tfrac{2}{|J|} \pbc_J\big(\\
    & \qquad
       (|J|+1) D_{c_{j\bullet}} D_{k_{\bullet j}}\\
    & \qquad + \textstyle\sum_{j' \in J} D_{c_{j'\bullet}} D_{k_{jj'}}\\
    & \qquad - \textstyle\sum_{j' \in J} D_{c_{j'\bullet}} D_{k_{\bullet j'}}
    \big) R_1 \pbc_K(X_\bullet)
   \modp \displaybreak[0] \\
%-----------
\Dw_2 \rho_2(S_1 U^2) & = 
   -\tfrac{2(|K|-1)}{|I|} R_2 \pbc_J(C_\bullet)
   +2R_2 \pbc_J(C_j)\\
   & \qquad
   +\tfrac{2(|K|-1)}{|I|-2} \pbc_J(D_{k_{\bullet j}}) R_2 \pbc_J(Z)
   +\tfrac{2(|K|+1)}{|I|} R_2 \pbc_K(C_\bullet) \\
   & \qquad
   -\tfrac{2(|J|-1)}{|I|-2} \pbc_J(D_{k_{\bullet j}}) R_2 \pbc_K(Z)
    \modp
\end{align*}
\end{itemize}
These Weyl algebra identities are by direct calculation;
this is algorithmically straightforward,
best done using symbolic computation.
They imply \eqref{eq:susu} hence Lemma \ref{lemma:key}.
\qed\end{proof}
%%%%%%%%%%%%%%%%%%%%%%%%%%%%%%%%%%%%%%%%%%%%%%%%%%%%%%%%%%%%%

\begin{proof}[of Theorem \ref{theorem:main}]
The proof is by induction on $|I|$.
The base case $|I|=3$ is by direct calculation. As an example,
if $I = \{1,2,3\}$ then 
\[
      \widetilde{C}_1 M_I = \widetilde{C}_2 M_I = \widetilde{C}_3 M_I = 
\text{const}\cdot e_{12}e_{23}e_{31}
\] 
It follows that $C_i M_I = 0$ for $i=1,2,3$ as required.
%---------------------------------------------------
Let now $|I| \geq 4$. Let $f = A_I$ (YM) or $f = M_I$ (GR).
For every $I = J \sqcup K$ with $|J|, |K| \geq 2$ write
\[
f = (f-g) + g
\qquad
g = \begin{cases}
\text{const}\cdot
 U \frac{\down^\ast_J(A_\JB) \down^\ast_K(A_\KB)}{\xi} & \text{for YM}\\
\text{const}\cdot
 U^2 \frac{\down^\ast_J(M_\JB) \down^\ast_K(M_\KB)}{\xi} & \text{for GR}
\end{cases}
\]
Note that:
\begin{itemize}
\item By the  recursion in
Definition \ref{def:rec},
the difference $f-g$ does not have a pole along $\xi = 0$
and therefore neither do $A_i(f-g)$, $B_i(f-g)$, $C_i(f-g)$.
\item By the induction hypothesis,
we can invoke Lemma \ref{lemma:key} and conclude that 
also $A_ig$, $B_ig$, $C_ig$ do not have a pole along $\xi = 0$.
\end{itemize}
Hence
$A_i f$, $B_i f$, $C_i f$
have no pole along $\xi = 0$ for every decomposition $I = J \sqcup K$,
and therefore (by Hartogs extension as in the proof of Lemma \ref{lemma:uniq})
we have
\[
A_i f,\; B_if,\; C_i f \in R_{C(I)}
\]
where $R_{C(I)}$ is the ring of polynomials.
To show that they are actually zero:
\begin{itemize}
\item Use \eqref{eq:acom} and Lemma \ref{lemma:vanish} (with $u = A_if$)
      to conclude that $A_i f = 0$.
\item Then use \eqref{eq:bcom} and Lemma \ref{lemma:vanish} (with $u = B_i f$)
      to conclude that $B_i f = 0$.
\item Then use \eqref{eq:ccom} and Lemma \ref{lemma:vanish} (with $u = C_i f$)
      to conclude that $C_i f = 0$.
\end{itemize}
\qed\end{proof}
%---------------------------------------------------------------------------

\section{Acknowledgments}
M.R.~is grateful to have received funding from ERC through grant agreement No.~669655.

\appendix
\section{Some tree amplitudes}\label{app:formulas}

The expressions below are rational functions
on the vector space $C(I)$
in Definition \ref{def:cvecs},
defined using the relations \eqref{eq:linrel}.
Not all symmetries are immediate from the formulas, for example
the permutation invariance of the GR amplitudes is not.
The expressions are up to an overall multiplicative constant.
\begin{itemize}
\item The YM amplitude $A_I$ for $I = \{1,2,3\}$: $-c_{23} e_{12}+c_{21} e_{23}-c_{12} e_{31}$.
\item The GR amplitude $M_I$ for $I = \{1,2,3\}$: $(-c_{23} e_{12}+c_{21} e_{23}-c_{12} e_{31})^2$.
\item The YM amplitude $A_I$ for $I = \{1,2,3,4\}$, with canonical cyclic order:
{\footnotesize
\begin{multline*}
-\frac{c_{14} c_{23} e_{12}}{k_{12}}+\frac{c_{13} c_{24} e_{12}}{k_{12}}
-\frac{c_{14} c_{23} e_{12}}{k_{23}}-\frac{c_{12} c_{14} e_{13}}{k_{12}}  
-\frac{c_{12} c_{24} e_{13}}{k_{12}}+\frac{c_{12} c_{13} e_{14}}{k_{12}}
+\frac{c_{12} c_{23} e_{14}}{k_{12}}+\frac{c_{14} c_{21} e_{23}}{k_{12}}  \displaybreak[0] \\
+\frac{c_{21} c_{24} e_{23}}{k_{12}}-\frac{c_{13} c_{21} e_{24}}{k_{12}}
-\frac{c_{21} c_{23} e_{24}}{k_{12}}-\frac{c_{12} c_{31} e_{34}}{k_{12}} 
+\frac{c_{21} c_{32} e_{34}}{k_{12}}+\frac{c_{14} c_{32} e_{13}}{k_{23}}
+\frac{c_{12} c_{23} e_{14}}{k_{23}}-\frac{c_{13} c_{32} e_{14}}{k_{23}}  \displaybreak[0] \\
+\frac{c_{14} c_{21} e_{23}}{k_{23}}+\frac{c_{21} c_{24} e_{23}}{k_{23}}
+\frac{c_{24} c_{31} e_{23}}{k_{23}}-\frac{c_{21} c_{23} e_{24}}{k_{23}} 
-\frac{c_{23} c_{31} e_{24}}{k_{23}}+\frac{c_{21} c_{32} e_{34}}{k_{23}}
+\frac{c_{31} c_{32} e_{34}}{k_{23}}-\frac{e_{34} e_{12} k_{23}}{k_{12}}\\
-\frac{e_{14} e_{23} k_{12}}{k_{23}}-e_{34} e_{12}-e_{14} e_{23}+e_{13} e_{24}
\end{multline*}}
\item The GR amplitude $M_I$ for $I = \{1,2,3,4\}$, as an unordered set:
{\footnotesize\begin{multline*}
k_{32}
\Big(
-e_{34} e_{12}-\frac{e_{34} k_{23} e_{12}}{k_{12}}-\frac{c_{14} c_{23} e_{12}}{k_{12}}
+\frac{c_{13} c_{24} e_{12}}{k_{12}}-\frac{c_{14} c_{23} e_{12}}{k_{23}}-e_{14} e_{23}
+e_{13} e_{24} \displaybreak[0] \\
-\frac{c_{12} c_{14} e_{13}}{k_{12}}-\frac{c_{12} c_{24} e_{13}}{k_{12}}
+\frac{c_{12} c_{13} e_{14}}{k_{12}}+\frac{c_{12} c_{23} e_{14}}{k_{12}} 
+\frac{c_{14} c_{21} e_{23}}{k_{12}}+\frac{c_{21} c_{24} e_{23}}{k_{12}}
-\frac{c_{13} c_{21} e_{24}}{k_{12}}-\frac{c_{21} c_{23} e_{24}}{k_{12}}  \displaybreak[0] \\
-\frac{c_{12} c_{31} e_{34}}{k_{12}}+\frac{c_{21} c_{32} e_{34}}{k_{12}}
+\frac{c_{14} c_{32} e_{13}}{k_{23}}+\frac{c_{12} c_{23} e_{14}}{k_{23}} 
-\frac{c_{13} c_{32} e_{14}}{k_{23}}+\frac{c_{14} c_{21} e_{23}}{k_{23}}
+\frac{c_{21} c_{24} e_{23}}{k_{23}}+\frac{c_{24} c_{31} e_{23}}{k_{23}}  \displaybreak[0] \\
-\frac{c_{21} c_{23} e_{24}}{k_{23}}-\frac{c_{23} c_{31} e_{24}}{k_{23}}
+\frac{c_{21} c_{32} e_{34}}{k_{23}}+\frac{c_{31} c_{32} e_{34}}{k_{23}}
-\frac{e_{14} e_{23} k_{12}}{k_{23}}\Big)\cdot   \displaybreak[0] \\
\cdot \Big(
e_{34} e_{12}-\frac{c_{13} c_{14} e_{12}}{k_{13}}-\frac{c_{13} c_{34} e_{12}}{k_{13}}
+\frac{c_{14} c_{23} e_{12}}{k_{32}}-e_{13} e_{24}-e_{14} e_{32}  \displaybreak[0] \\ 
-\frac{e_{13} e_{24} k_{32}}{k_{13}}-\frac{c_{14} c_{32} e_{13}}{k_{13}}
+\frac{c_{12} c_{34} e_{13}}{k_{13}}+\frac{c_{12} c_{13} e_{14}}{k_{13}} 
+\frac{c_{13} c_{32} e_{14}}{k_{13}}-\frac{c_{13} c_{21} e_{24}}{k_{13}}
+\frac{c_{23} c_{31} e_{24}}{k_{13}}+\frac{c_{14} c_{31} e_{32}}{k_{13}}  \displaybreak[0] \\
+\frac{c_{31} c_{34} e_{32}}{k_{13}}-\frac{c_{12} c_{31} e_{34}}{k_{13}}
-\frac{c_{31} c_{32} e_{34}}{k_{13}}-\frac{c_{14} c_{32} e_{13}}{k_{32}}  
-\frac{c_{12} c_{23} e_{14}}{k_{32}}+\frac{c_{13} c_{32} e_{14}}{k_{32}}
+\frac{c_{21} c_{23} e_{24}}{k_{32}}+\frac{c_{23} c_{31} e_{24}}{k_{32}}  \displaybreak[0] \\
+\frac{c_{14} c_{31} e_{32}}{k_{32}}+\frac{c_{21} c_{34} e_{32}}{k_{32}}
+\frac{c_{31} c_{34} e_{32}}{k_{32}}-\frac{c_{21} c_{32} e_{34}}{k_{32}}  
-\frac{c_{31} c_{32} e_{34}}{k_{32}}-\frac{e_{14} e_{32} k_{13}}{k_{32}}
\Big)
\end{multline*}}
\end{itemize}
The dimension-neutral YM amplitudes are discussed e.g.~in \cite{bbdf},
explicit expressions are at
{\footnotesize \url{www.damtp.cam.ac.uk/user/crm66/SYM/pss.html}}.
To get the dimension-neutral GR amplitudes one can use the
Kawai-Lewellen-Tye or KLT relations, for a review see \cite{sonder}.

%%%%%%%%%%%%%%%%%%%%%%%%%%%%%%%%%%%%%%%%%%%%%%%%%%%%%%%%%%%%%%%%

\section{Mathematica computer code}\label{app:code}

Here we provide definitions from Theorem \ref{theorem:main}
and Lemma \ref{lemma:dops} as Wolfram Mathematica computer code.
The code is rudimentary and not practical except for very low $n$.
\begin{verbatim}
(* Set number of legs, n>=3 *)
n=3;

(* Linear relations *)
relations=Flatten[{Table[{k[i,i],c[i,i],e[i,i]},{i,1,n}],
                   Table[{k[i,j]-k[j,i],e[i,j]-e[j,i]},{i,1,n},{j,1,n}],
                   Table[Sum[{k[i,j],c[i,j]},{i,1,n}],{j,1,n}]}];
normalForm=First[Solve[Thread[relations==0]]];

(* Definition of D operators *)
z=KroneckerDelta;
Dk[i_,j_][f_]:=(1-z[i,j])*Sum[(1-z[a,b])*(z[i,a]*z[j,b]+z[i,b]*z[j,a]
    -1/(n-2)*(z[i,a]+z[i,b]+z[j,a]+z[j,b])+2/((n-1)*(n-2)))*D[f,k[a,b]],{a,1,n},{b,1,n}];
Dc[i_,j_][f_]:=(1-z[i,j])*Sum[(1-z[a,j])*(z[i,a]-1/(n-1))*D[f,c[a,j]],{a,1,n}];
De[i_,j_][f_]:=(1-z[i,j])*(D[f,e[i,j]]+D[f,e[j,i]]);

(* Definition of A,B,C operators *)
AOp[i_][f_]:=Sum[1/2*k[p,q]*Dc[p,i][Dc[q,i][f]]+c[p,q]*Dc[p,i][De[q,i][f]]
    +1/2*e[p,q]*De[p,i][De[q,i][f]],{p,1,n},{q,1,n}];
BOp[i_][f_]:=Sum[c[p,q]*Dk[i,p][De[i,q][f]]+e[p,q]*Dc[i,p][De[i,q][f]]
    +k[p,q]*Dk[i,p][Dc[q,i][f]]+c[q,p]*Dc[i,p][Dc[q,i][f]]
    -c[q,p]*Dk[p,q][De[p,i][f]]-e[p,q]*Dc[p,q][De[p,i][f]]
    -k[p,q]*Dk[p,q][Dc[p,i][f]]-c[p,q]*Dc[p,i][Dc[p,q][f]],{p,1,n},{q,1,n}];
COpTilde[i_][f_]:=Sum[1/2*e[p,q]*Dc[i,p][Dc[i,q][f]]+c[q,p]*Dk[i,q][Dc[i,p][f]]
    +1/2*k[p,q]*Dk[i,p][Dk[i,q][f]]-e[p,q]*Dc[i,p][Dc[p,q][f]]
    -c[q,p]*Dk[p,q][Dc[i,p][f]]-c[p,q]*Dk[p,i][Dc[p,q][f]]
    -k[p,q]*Dk[p,i][Dk[p,q][f]],{p,1,n},{q,1,n}];
COp[i_][f_]:=COpTilde[i][f]-1/n*Sum[COpTilde[p][f],{p,1,n}];

(* Example: C_1 annihilates the GR amplitude with n=3 *)
M3=(-c[2,3]*e[1,2]+c[2,1]*e[2,3]-c[1,2]*e[3,1])^2;
COp[1][M3] /. normalForm // Simplify (* yields 0 *)
\end{verbatim}

\section{Additional formulas}\label{app:kref}

This is an extension of Section \ref{sec:kinematic}
that provides definitions and formulas that are useful in the proof
of the technical Lemma \ref{lemma:key}.
We extend \eqref{mdiag} to four linear maps:
\[
\begin{tikzcd}[column sep = huge]
& M \arrow[dl,shift right = 1.5, swap,"\down_J"]
    \arrow[dr,shift left = 1.5, "\down_K"] & \\
C(\JB)
\arrow[ur,shift right = 1.5, swap, "\up_J"]
& &
\arrow[ul, shift left = 1.5, "\up_K"]
C(\KB)
\end{tikzcd}
\]
The maps $\up_J$ and $\up_K$ are explicit right-inverses
of $\down_J$ and $\down_K$, 
so $\down_J \up_J = \mathbbm{1}$ and $\down_K \up_K = \mathbbm{1}$.
Hence $\up_J \down_J$ and $\up_K \down_K$ are projections $M \to M$.
The direction $\Xi$ is in the kernel of these projections.
We choose $\up_J$ and $\up_K$ so that their
images are contained in $\Mp$ and so that
$\down_K \up_J = 0$ and $\down_J \up_K = 0$.
The result of this is a useful decomposition
\begin{equation}\label{eq:decomp}
M = \C \Xi \oplus \underbrace{\image \pi \oplus
\underbrace{\image \up_J}_{\simeq C(\JB)}
\oplus
\underbrace{\image \up_K}_{\simeq C(\KB)}}_{= \Mp}
\end{equation}
where $\pi = \mathbbm{1}_{\Mp} - \up_J \down_J - \up_K \down_K$
is a projection, $\pi^2 = \pi$,
with $\down_J \pi = \down_K \pi = 0$.
%--------------------------------------------
\begin{lemma} \label{lemma:pushf}
The map $\down_J$ (see Lemma \ref{lemma:pb}) satisfies,
and is equivalently defined by,
\begin{align*}
\Dp_{k_{jj'}}
& \mapsto D_{k_{jj'}}
- \textstyle\frac{(|K|-1)(1-\delta_{jj'})}{|I|-2} (D_{k_{j\bullet}} + D_{k_{j' \bullet}}) &
\Dp_{e_{jj'}}
& \mapsto D_{e_{jj'}}\\
%----
\Dp_{k_{jk}}, \Dp_{k_{kj}}
& \mapsto \textstyle\frac{|J|-1}{|I|-2} D_{k_{j\bullet}} &
\Dp_{e_{jk}}, \Dp_{e_{kj}}
& \mapsto 0\\
%----
\Dp_{k_{kk'}} & \mapsto 0 &
\Dp_{e_{kk'}} & \mapsto 0\\
%----
\Dp_{c_{jj'}}
& \mapsto D_{c_{jj'}}
- \textstyle\frac{(|K|-1)(1-\delta_{jj'})}{|I|-1} D_{c_{\bullet j'}} &
D_{c_{j\bullet}} & \mapsto D_{c_{j\bullet}} \displaybreak[0]\\
%---
\Dp_{c_{jk}} & \mapsto 0 &
D_{e_{j\bullet}} & \mapsto D_{e_{j\bullet}}\\
%---
\Dp_{c_{kj}} & \mapsto \textstyle\frac{|J|}{|I|-1} D_{c_{\bullet j}} &
D_{c_{k\bullet}} & \mapsto 0\\
%----
\Dp_{c_{kk'}} & \mapsto 0 &
D_{e_{k\bullet}} & \mapsto 0\\
%----
&&\Xi & \mapsto 0
\end{align*}
for all $j,j' \in J$ and $k,k' \in K$.
The elements on the left are in $M \subset D_M$.
Those on the right are in $C(\JB) \subset D_{C(\JB)}$,
using Lemma \ref{lemma:dops} for $\JB$.
Analogous for $\down_K$.
\end{lemma}
\begin{proof}
Let $X \mapsto Y$ be any of these claimed assignments.
We must check that $Y(y) = X(\down_J^\ast y)$ for all $y \in C(\JB)^\ast$
where here $Y(y)$ means applying $Y$ to $y$ as a function\footnote{%
With respect to bases and their dual bases, the matrix for $\down_J$
is the transpose of the matrix for $\down_J^\ast$ defined in Lemma \ref{lemma:pb},
but we do not work with bases.}.
Here $\down_J^\ast$ is given by Lemma \ref{lemma:pb}. 
Consider for example the case $X = \Dp_{k_{jk}}$ and $y = k_{ab} \in C(\JB)^\ast$ with $a,b \in J$.
Then
\begin{multline*}
X(\down_J^\ast(y)) = \Dp_{k_{jk}}(\kp_{ab})
= \Dp_{k_{jk}}(k_{ab})
= D_{k_{jk}}(k_{ab}) + \textstyle\frac{(|K|-1)(|J|-1)}{(|I|-1)(|I|-2)} \Xi(k_{ab})\\
                   = \textstyle(1-\delta_{ab}) (-\frac{1}{|I|-2}(\delta_{ja} + \delta_{jb})
                       + \frac{2}{(|I|-1)(|I|-2)})
+ \textstyle\frac{(|K|-1)(|J|-1)}{(|I|-1)(|I|-2)} \frac{2(1-\delta_{ab})}{|J|(|J|-1)}
\end{multline*}
using Lemmas \ref{lemma:dops} and \ref{lemma:dpx}. Similarly,
\[
Y(y) = \textstyle\frac{|J|-1}{|I|-2} D_{k_{j\bullet}}(k_{ab})
= \textstyle\frac{|J|-1}{|I|-2}(1-\delta_{ab})(
     -\frac{1}{|J|-1}(\delta_{ja} + \delta_{jb}) + \frac{2}{|J|(|J|-1)})
\]
using Lemma \ref{lemma:dops} for $\JB$. 
In this example we see that
$X(\down_J^\ast(y)) = Y(y)$.
\qed\end{proof}
%%%%%%%%%%%%%%%%%%%%%%%%%%%%%%%%%%%%%%%%%%%%%%%%%%%%%%%%%%%%%%%%%%%%%%%%%%%%%%%%%%%%%%%%%%
\begin{lemma}[Right-inverse] \label{lemma:rinvp}
The map $\down_J$ is surjective. There exists a unique right inverse
$\up_J: C(\JB) \to M$ whose image is the subspace spanned by:
\begin{itemize}
\item All $\Dp_{k_{jj'}}$, $\Dp_{c_{jj'}}$, $\Dp_{e_{jj'}}$ with $j,j' \in J$.
\item All $D_{c_{j\bullet}}$, $D_{e_{j\bullet}}$ with $j \in J$.
\end{itemize}
This right-inverse $\up_J$ maps
\begin{align*}
D_{k_{jj'}} & \mapsto 
\Dp_{k_{jj'}}
-
\textstyle\frac{(|K|-1)(1-\delta_{jj'})}{|K|(|J|-1)}
\sum_{j'' \in J} (\Dp_{k_{jj''}} + \Dp_{k_{j'j''}})
\\
%------------------------
D_{k_{j\bullet}}, D_{k_{\bullet j}} & \mapsto
- \textstyle\frac{|I|-2}{|K|(|J|-1)}\sum_{j'\in J}\Dp_{k_{jj'}}\\
D_{c_{jj'}} & \mapsto 
\Dp_{c_{jj'}}
-
 \textstyle\frac{(|K|-1)(1-\delta_{jj'})}{|J||K|} \sum_{j'' \in J} \Dp_{c_{j''j'}}
\displaybreak[0]\\
%-------------------------
D_{c_{\bullet j}} & \mapsto - \textstyle\frac{|I|-1}{|J||K|}\sum_{j'\in J} \Dp_{c_{j'j}}\\
D_{c_{j\bullet}} & \mapsto D_{c_{j\bullet}} \\ 
D_{e_{jj'}} & \mapsto \Dp_{e_{jj'}}\\
D_{e_{j\bullet}}, D_{e_{\bullet j}} & \mapsto D_{e_{j\bullet}}
\end{align*}
for all $j,j' \in J$.
We have $\down_K \up_J = 0$. Analogous for $\up_K$ with $\down_J \up_K = 0$.
\end{lemma}
\begin{proof}
First note that these assignments define a map $\up_J$
because we have specified it on a set that spans $C(\JB)$
and it respects all relations,
for instance $\sum_{j' \in J} D_{c_{j'j}} + D_{c_{\bullet j}}$
is zero in $C(\JB)$ and must be mapped to zero, which it is.
Now $\down_J \up_J = \mathbbm{1}$
is by direct calculation using Lemma \ref{lemma:pushf},
in particular $\down_J$ is surjective and $\up_J$ is injective.
The image of $\up_J$ is as claimed.
The image determines the right-inverse, so it is unique as claimed.
\qed\end{proof}
%----------------------------------
\begin{corollary}\label{corollary:decomp}
The internal direct sum decomposition \eqref{eq:decomp} holds. Also:
\begin{itemize}
\item The extension space $E \subset \Mp$ is contained in
$\image \up_J \oplus \image \up_K$.
\item The subspace $\image \pi \subset \Mp$ is the span of all
\begin{align*}
\Dw_{k_{jk}} & = \Dp_{k_{jk}} - \textstyle\frac{1}{|J|}\sum_{j' \in J} \Dp_{k_{j'k}}\\
             & \qquad\qquad
               - \textstyle\frac{1}{|K|}\sum_{k' \in K} \Dp_{k_{jk'}}
               + \textstyle\frac{1}{|J||K|}\sum_{j' \in J, k' \in K} \Dp_{k_{j'k'}}\\
\Dw_{c_{jk}} & = \Dp_{c_{jk}} - \textstyle\frac{1}{|J|}\sum_{j' \in J} \Dp_{c_{j'k}}\\
\Dw_{c_{kj}} & = \Dp_{c_{kj}} - \textstyle\frac{1}{|K|}\sum_{k' \in K} \Dp_{c_{k'j}}\\
\Dw_{e_{jk}} & = \Dp_{e_{jk}}
\end{align*}
with $j \in J$, $k \in K$. Its dimension is $(2|J|-1)(2|K|-1)$.
In $D_M$ its elements commute with all elements in the images of
$\down_J^\ast$, $\down_K^\ast$ in particular
\[
\kp_{jj'}, \cp_{jj'}, \ep_{jj'},
\kp_{kk'}, \cp_{kk'}, \ep_{kk'},
c_{j\bullet}, e_{j\bullet},
c_{k\bullet}, e_{k\bullet}
\;\;\;
\in\;\;\; M^\ast
\]
for all $j,j' \in J$ and $k,k' \in K$.
\end{itemize}
\end{corollary}

%-----------------------------------------------------------------
\newcommand{\nonlocop}{\tc{hcolor}{Y}}
\newcommand{\locop}{\nonlocop_{\text{loc}}}
\newcommand{\proj}{{\tc{qcolor}{\pi}}}
\section{Relation to Loebbert, Mojaza, Plefka \cite{lmp}} \label{app:translation}
The paper \cite{lmp} investigates a potential hidden conformal symmetry
for GR amplitudes in general spacetime dimension,
and provides conceptual backing relating to soft dilatons.
No attempt is made to review these aspects here.
The purpose of this appendix is merely to 
clarify the relation between
the conjecture for GR amplitudes in \cite{lmp}
and the new formulation in terms of the differential operators $B_i$, $C_i$
in Theorem \ref{theorem:main}.
(Here we do not discuss the more tentative conjecture for YM amplitudes in \cite{lmp}.
Nor do we discuss the operators $A_i$
that are not part of the conjecture in \cite{lmp}.)
The notation used in this appendix uses notation from \cite{lmp},
notation from this paper,
and additional notation used only in this appendix.
\begin{itemize}
\item The conjecture we are referring to are equations (5.12) and (5.13) in \cite{lmp}
for general $n \geq 3$. The cases $n=3,4,5,6$ are separately discussed in \cite{lmp}.
\item 
In \cite{lmp},
the primary differential operators are certain conformal generators.
The special conformal generators
are expressed in terms of operators\footnote{%
Actually in \cite{lmp} the notation $F_i$, $G_i$ is the result
of applying certain operators to the amplitudes,
but in this appendix we use this notation for the differential operators themselves.}
$F_i$, $G_i$ in \cite{lmp}
closely related to $B_i$, $C_i$ in this paper respectively.
In fact, superficially,
the formulas in Appendix A of \cite{lmp}
coincide with the right hand sides of
\eqref{eqs:b}, \eqref{eqs:c}.
The symbols entering the formulas
do however not have exactly the same meaning
as we now discuss.
\item
In \cite{lmp},
ordinary partial derivatives
$\p_{s_{ij}}$, $\p_{w_{ij}}$, $\p_{e_{ij}}$
(the variables $s,w,e$ in \cite{lmp}
correspond to $k,c,e$ in this paper respectively)
on the ambient vector space are used
that are not compatible with the relations \eqref{eq:linrel}.
Let us refer to $A$ as the ambient vector space and to
$C \subset A$ as the subspace given by \eqref{eq:linrel}.
It is called `constraint surface' in \cite{lmp}.
Note that $F_i, G_i: \Frac(R_A) \to \Frac(R_A)$.
The conjecture in \cite{lmp}
is that the GR amplitudes
are annihilated by $\nonlocop(F_i)$, $\nonlocop(G_i)$ where,
in notation used only in this appendix,
\[
\nonlocop(X) = 
\textstyle\sum_{\proj} R X \proj E\;\;:\;\;\Frac(R_C) \to \Frac(R_C)
\]
where (all maps are linear):
\begin{itemize}
\item $E: \Frac(R_C) \to \Frac(R_A)$
extends functions to the ambient space,
it is the pullback along a specific projection $A \to A$ with image $C$.
This is referred to as `resolving the constraints' in \cite{lmp}
and for the specific choice of $E$ we refer to \cite{lmp}.
\item $\proj : \Frac(R_A) \to \Frac(R_A)$
is a permutation of the index set $I$,
and $\sum_\proj$ is a sum over all.
So $\sum_\proj \proj$ corresponds to symmetrization.
\item
$X : \Frac(R_A) \to \Frac(R_A)$ is one of the
partial differential operators of interest, $F_i$ or $G_i$ in the notation
of \cite{lmp},
on the ambient space.
\item $R : \Frac(R_A) \to \Frac(R_C)$ restricts functions
from $A$ to $C$\footnote{Actually $R$ is only partially defined
but we gloss over that since it is inconsequential here.}.
\end{itemize}
\item
Note that $\nonlocop(X)$ is not a partial differential operator on $C$ because
$\proj$ is not local.
But one can define a new, different operator $\locop(X)$ that
has the virtue of being local
in fact a polynomial differential operator on $C$; and that
nevertheless coincides with $\nonlocop(X)$ when acting on permutation invariant
elements of $\Frac(R_C)$. Namely\footnote{%
Here $\proj^{-1} : \Frac(R_C) \to \Frac(R_C)$.}
\[
\locop(X) = 
\textstyle\sum_{\proj} RX\proj E\proj^{-1} \;\;:\;\;\Frac(R_C) \to \Frac(R_C)
\]
Clearly $\nonlocop(X) f = \locop(X) f$ for all rational functions $f \in \Frac(R_C)$ that
are permutation invariant, $\proj f = f$ for all $\proj$.
But $\locop(X)$ is local,
because
$\locop(X) = \sum_\proj RXE_\proj$
where $E_\proj = \proj E\proj^{-1} : \Frac(R_C) \to \Frac(R_A)$
is simply the pullback along another projection
$A \to A$ with image $C$.
\item
The GR amplitudes are permutation invariant.
Therefore the conjecture in \cite{lmp} that $\nonlocop(F_i)$, $\nonlocop(G_i)$
annihilate the
GR amplitude is equivalent to $\locop(F_i)$, $\locop(G_i)$ annihilating the GR amplitude.
\item
The $B_i$, $C_i$
in \eqref{eqs:b}, \eqref{eqs:c} are,
up to normalization, the operators $\locop(F_i)$, $\locop(G_i)$
but presented directly as differential operators
on $C = C(I)$.
A detailed translation
to \eqref{eqs:b}, \eqref{eqs:c} is omitted.
There are terms proportional to the conformal dimension $\Delta$
in $F_i$, $G_i$, see \cite{lmp}, but they do not contribute
to $\locop(F_i)$, $\locop(G_i)$ which therefore are independent of $\Delta$.
\end{itemize}

%%%%%%%%%%%%%%%%%%%%%%%%%%%%%%%%%%%%%%%%%%%%%%%%%%%%%%%%%%%%%%%%%%%%%%%%
\end{document}